\documentclass[12pt,reqno]{amsart}
\usepackage{lmodern}
\usepackage[T1]{fontenc}

\usepackage[numbers]{natbib}
\usepackage{amsmath,amssymb,latexsym} 
\usepackage{fullpage}
\usepackage{etoolbox}
\usepackage{enumitem}
\usepackage{color}
\usepackage[colorlinks,linkcolor=blue,citecolor=magenta]{hyperref}
\usepackage[colorinlistoftodos,bordercolor=orange,backgroundcolor=orange!20,linecolor=orange,textsize=scriptsize]{todonotes}
\usepackage{soul}
\usepackage{microtype}
\usepackage{yhmath}
\usepackage{float}
\usepackage{algorithm}
\usepackage[noend]{algpseudocode}
\usepackage{bbold}
\usepackage{booktabs}
\usepackage{comment}
\usepackage{mathtools}
\usepackage{calc}
\usepackage{float}
\usepackage{pstricks}
\usepackage{pstricks-add}

\usepackage{geometry}
\newtheorem{theorem}{Theorem}

\newtheorem{proposition}{Proposition}

\theoremstyle{definition}

\newcommand{\Real}{{\mathbb R}}

\newcommand{\np}{\bar{n}}

\DeclareMathOperator*{\argmax}{arg\,max}

\DeclareMathOperator*{\argmin}{arg\,min}
\DeclareFontFamily{U}{mathx}{}
\DeclareFontShape{U}{mathx}{m}{n}{ <-> mathx10 }{}
\DeclareSymbolFont{mathx}{U}{mathx}{m}{n}
\DeclareFontSubstitution{U}{mathx}{m}{n}
\DeclareMathAccent{\widecheck}{0}{mathx}{"71}

\newcommand{\setap}{\widecheck{A}}
\newcommand{\mFA}{\text{F1}}
\newcommand{\mFB}{\text{F2}}
\newcommand{\incoming}[1]{\delta^-(#1)}
\newcommand{\outgoing}[1]{\delta^+(#1)}

\newcommand{\prob}{\texttt{MPC-ARC}}
\newcommand{\fprob}{\texttt{MFC-ARC}}
\newcommand{\polyt}{\mathcal{P}_{\fprob}}
\newcommand{\bmsection}[1]{\section*{\textbf{#1}}}

\newcommand{\codef}[1]{\texttt{#1}}

\title{On a Variant of the Minimum Path Cover Problem in Acyclic Digraphs: Computational Complexity Results and Exact Method}

\author{Nour ElHouda Tellache}
\address{N.E.H. Tellache, Decision Support \& Operations Research Group, Department of Informatics, University of Fribourg, Fribourg, Switzerland}
\email{nourelhouda.tellache@unifr.ch}

\author{Roberto Baldacci}
\address{R. Baldacci, College of Science and Engineering, Hamad Bin Khalifa University, Qatar Foundation, Doha, Qatar}
\email{rbaldacci@hbku.edu.qa}

\keywords{constrained minimum path cover; directed acyclic graphs; complexity; valid inequalities; separation problems; branch-and-cut}

\begin{document}
	
	\maketitle
	
	\begin{abstract}
		The Minimum Path Cover (\texttt{MPC}) problem consists of finding a minimum-cardinality set of node-disjoint paths that cover all nodes in a given graph. We explore a variant of the \texttt{MPC} problem on acyclic digraphs (DAGs) where, given a subset of arcs, each path within the \texttt{MPC} should contain at least one arc from this subset. We prove that the feasibility problem is strongly $\mathcal{NP}$-hard on arbitrary DAGs, but the problem can be solved in polynomial time when the DAG is the transitive closure of a path.
		Given that the problem may not always be feasible, our solution focuses on covering a maximum number of nodes with a minimum number of node-disjoint paths, such that each path includes at least one arc from the predefined subset of arcs. This paper introduces and investigates two integer programming formulations for this problem. We propose several valid inequalities to enhance the linear programming relaxations, employing them as cutting planes in a branch-and-cut approach. The procedure is implemented and tested on a wide range of instances, including real-world instances derived from an airline crew scheduling problem, demonstrating the effectiveness of the proposed approach.
	\end{abstract}
\section{Introduction}
A \textit{path cover} of a graph refers to a set of node-disjoint paths that collectively cover all the nodes of the graph. The Minimum Path Cover (\texttt{MPC}) problem involves determining a path cover with the smallest possible cardinality. The \texttt{MPC} problem and its various variants have been extensively studied due to their theoretical significance and their wide-ranging applications in fields such as program testing~\cite{Ntafos1979}, code optimization~\cite{Boesch1977}, scheduling~\cite{Colbourn1985,zhan2016,tellache2017}, bioinformatics~\cite{Eriksson2008, Bao2013, song2013, Rizzi2014, Chang2015}, and many others.

The \texttt{MPC} problem is $\mathcal{NP}$-hard on arbitrary graphs since it generalizes the Hamiltonian path problem~\cite{GJ79}. $\mathcal{NP}$-hardness results have been established for specific graph classes, including planar cubic triply-connected graphs~\cite{Garey1976}, chordal bipartite and chordal split graphs~\cite{MULLER1996}, and circle graphs~\cite{DAMASCHKE1989}. However, polynomial-time algorithms exist for trees~\cite{boesch1974}, interval graphs~\cite{ARIKATI1990}, cographs~\cite{lin1995}, block graphs~\cite{chang2002}, and cocomparability graphs~\cite{DCBD13}.
One noteworthy positive result dates back to the 1950s and applies to acyclic digraphs (DAGs). This polynomial algorithm is based on the work of Dilworth~\cite{dilworth1950} and Fulkerson~\cite{Fulkerson1956} and involves finding a maximum matching in a bipartite graph constructed from the given DAG.

We study the following variant of the \texttt{MPC} problem on a DAG (denoted \prob): given a DAG $G=(V,A)$ with node set $V=\{1,2,\dots,n\}$ and a subset $\setap$ of the arc set $A$, we seek an \texttt{MPC} such that every path in the path cover has at least one arc from $\setap$. A path is \textit{feasible} if it has an arc from $\setap$. We denote the problem of checking whether a graph has a path cover of feasible paths by \texttt{FPC-ARC}. Since \texttt{MPC-ARC} may not always be feasible, we focus on solutions to the problem of maximizing the number of covered nodes using a minimum number of feasible paths, which we refer to as \fprob.

The \texttt{MPC-ARC} problem arises naturally in various contexts, including the airline crew scheduling problem~\cite{Gopalakrishnan2005, tellache2024linear}. Some scheduling systems aim to construct pilot schedules that minimize the number of pilots required to cover all flights. A schedule consists of a sequence of pairings (i.e., sequences of flights starting and ending at the same crew base) that must comply with airline regulations, collective agreements, and pre-assigned duties. Among these constraints is the requirement that each schedule includes a rest block of at least $r$ consecutive days off. We can associate a graph with this problem, where each node represents a pairing, and an arc exists between two pairings if they can be operated sequentially. The resulting graph is acyclic, and if the arcs spanning more than $r$ days are identified as the set $\setap$, then a feasible schedule for a pilot corresponds to a path in the graph that includes at least one arc from $\setap$. Note that although crew scheduling is subject to numerous constraints, analyzing the complexity of the problem with this specific rest-block constraint provides insight into the minimal components that contribute to the hardness of the crew scheduling problem.

Several constrained path problems, with or without the requirement to cover all nodes in DAGs have been investigated in the literature. Ntafos and Hakimi~\cite{Ntafos1979} demonstrated that finding a minimum-cardinality set of paths covering ``required pairs'' of nodes (i.e., ensuring that for each required pair, at least one path contains both nodes of the pair) is  $\mathcal{NP}$-hard. They also presented a polynomial-time algorithm for the case of ``required paths'' (i.e., ensuring that for each required path, at least one path contains the required path as a subpath). Note that neither of these problems requires covering all the nodes in the DAG. The authors also introduced the concept of ``must paths''. They proved that finding a minimum-cardinality set of paths covering all nodes while respecting must-path constraints (i.e., for each must path originating from node $i$, a path visiting $i$ should contain the must path as a subpath) can be solved in polynomial time. Finally, they examined the problem of finding a minimum-cardinality set of paths covering all nodes while adhering to ``impossible paths'' constraints (i.e., ensuring no path contains any impossible path) and showed that this problem is \(\mathcal{NP}\)-hard.
Beerenwinkel et al.\cite{Beerenwinkel2015} explored the ``required pairs'' problem introduced in~\cite{Ntafos1979} and analyzed the complexity based on a parameter characterizing the pairs called the ``overlapping degree''. They also investigated the same problem with the additional constraint that all nodes should be covered; they showed that it is $\mathcal{NP}$-hard to decide if there exists a solution consisting of at most three paths, and they gave a polynomial-time algorithm for computing a solution with at most two paths. In parallel to this work, Rizzi et al.\cite{Rizzi2014} studied two generalizations of the ``required paths'' and ``required pairs'' problems. They found that the generalization of the former problem and its weighted version can be solved in polynomial time. In generalizing the latter, which involves covering pairs of subpaths, the authors proved that even if the problem remains $\mathcal{NP}$-hard, it is fixed-parameter tractable in the total number of constraints. In all the mentioned works, the paths are not necessarily disjoint.

While other variants of constrained path problems on DAGs have been addressed in the literature and some algorithms have been proposed (see, e.g.,~\cite{song2013,Bao2013}), we have not identified any results that could inform the complexity of \texttt{MPC-ARC}.

The main contributions of this paper can be summarized as follows:
\begin{itemize} 
	\item We demonstrate that the \texttt{FPC-ARC} problem is strongly $\mathcal{NP}$-hard on arbitrary DAGs. However, when restricted to transitive closures of paths, \texttt{MPC-ARC} can be solved in polynomial time.
	\item We develop two mathematical formulations for \fprob. Notably, one formulation incorporates ``infeasible path elimination'' constraints.
	\item We derived several valid inequalities to strengthen the formulations and investigated the corresponding separation problems. The resulting separation algorithms are integrated into a branch-and-cut exact solution approach.
	\item We conducted extensive computational experiments on various instance sets to validate the effectiveness of the proposed exact method and the classes of valid inequalities.
\end{itemize}

The paper is organized as follows: The next section investigates the computational complexity of \texttt{MPC-ARC}. Section \ref{sec:mathforms} presents two mathematical formulations for \fprob, followed by Section \ref{sec:validineq} describing valid inequalities. Section \ref{sec:exactm} presents the exact algorithm. Section \ref{sec:compres} reports the results of an extensive computational analysis on various instance sets. Finally, Section \ref{sec:con} concludes the study and outlines future research directions. Additional details about the computational results are reported in the appendix.

\section{Complexity results} \label{CR}

In this section, we establish that the \texttt{FPC-ARC} problem is strongly $\mathcal{NP}$-hard on general DAGs. However, when restricted to the transitive closures of paths, the \texttt{MPC-ARC} problem becomes solvable in polynomial time. The following theorems formalize these results.

\begin{theorem}\label{FPCnp-hard}
	\texttt{FPC-ARC} is $\mathcal{NP}$-hard in the strong sense for arbitrary DAGs.
\end{theorem}
\begin{proof}
	We prove this result using a reduction from the 3-dimensional matching problem (\texttt{3-DM}), which is known to be $\mathcal{NP}$-complete in the strong sense (see, e.g., \cite{GJ79}). It is defined as follows. Given three disjoint sets $ X $, $ Y $ and $ Z $ each of size $ q $, and a set $ M \subseteq X\times Y \times Z $ of triples. Is there a subset $ M'\subseteq M $ of $ q $ triples such that each element of $ X\cup Y \cup Z $ is contained in exactly one of these triples? 
	
	Given an arbitrary instance of \texttt{3-DM}, we construct an instance of our problem as follows. Let $ X=\{x_{1},x_{2},\ldots,x_{q}\} $, $ Y=\{y_{1},y_{2},\ldots,y_{q}\} $ and $ Z=\{z_{1},z_{2},\ldots,z_{q}\} $. For each triple $m=(x_{i},y_{j},z_{k})\in M $, we construct the graph $G_{m}=(V_{m},A_{m})$, called a \textit{gadget}, depicted in Figure \ref{Gm}. The set $ V_{m} $ contains  4 ``internal'' nodes $ \{a_{1}^{m},a_{2}^{m}\}\cup \{b_{1}^{m},b_{2}^{m}\}$ that are adjacent to the nodes of $G_{m}$ only, and  3 ``external'' nodes $ x_{i} $, $ y_{j} $ and $z_{k}$ that can be shared with the other gadgets. The graph $G=(V,A) $ is obtained from the union of all the gadgets $G_{m} $ as follows:  $G=(V,A)=\bigg (\bigcup\limits_{m\in M}V_{m},\bigcup\limits_{m\in M} A_{m} \bigg) $ and the arcs of $\setap$ are the dashed arcs. Note that the graphs 	
	Given an arbitrary instance of \texttt{3-DM}, we construct an instance of our problem as follows. Let $ X=\{x_{1},x_{2},\ldots,x_{q}\} $, $ Y=\{y_{1},y_{2},\ldots,y_{q}\} $ and $ Z=\{z_{1},z_{2},\ldots,z_{q}\} $. For each triple $m=(x_{i},y_{j},z_{k})\in M $, we construct the graph $G_{m}=(V_{m},A_{m})$, called a \textit{gadget}, depicted in Figure \ref{Gm}. The set $ V_{m} $ contains  4 ``internal'' nodes $ \{a_{1}^{m},a_{2}^{m}\}\cup \{b_{1}^{m},b_{2}^{m}\}$ that are adjacent to the nodes of $G_{m}$ only, and  3 ``external'' nodes $ x_{i} $, $ y_{j} $ and $z_{k}$ that can be shared with the other gadgets. The graph $G=(V,A) $ is obtained from the union of all the gadgets $G_{m} $ as follows:  $G=(V,A)=\bigg (\bigcup\limits_{m\in M}V_{m},\bigcup\limits_{m\in M} A_{m} \bigg) $. The set of arcs $\setap$ consists of the dashed arcs within the gadgets $G_{m}$, and it is defined as $\setap = \bigcup\limits_{m\in M} \{(x_i, a_{1}^{m}), (a_{1}^{m},a_{2}^{m}), (b_{1}^{m},b_{2}^{m})\}$. Thus, $|\setap|=3|M|$. Note that each gadget $G_{m}$ is acyclic, and it can be verified that the union graph $G$ is also acyclic.
	
	\begin{figure}
		\begin{center}
			\psscalebox{1.0 1.0} 
			{
			\includegraphics{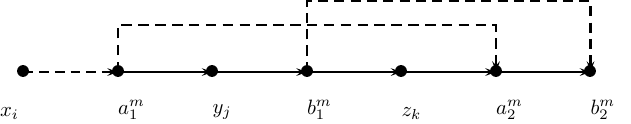}
			}
		\end{center}
		\caption{Gadget $ G_{m}=(V_{m},A_{m}) $}\label{Gm}
	\end{figure}
	
	We show that \texttt{3-DM} has a solution if and only if there exists a solution to \texttt{FPC-ARC}.
	
	\sloppy First assume that \texttt{3-DM} has a solution $ M' $. For each $ m'=(x_{i'},y_{j'},z_{k'}) \in M'$, the nodes of  $ G_{m'}$ are covered by a single path of the form  $(x_{i'}, a_{1}^{m'},y_{j'},b_{1}^{m'}, z_{k'},a_{2}^{m'}, b_{2}^{m'})$, this path is feasible since it includes the arc $(x_{i'}, a_{1}^{m'}) \in \setap$. To complete the path cover, we must account for the remaining nodes, the internal nodes of the $|M|-q$ triples of $\in M \setminus M'$. Let $m \in M \setminus M'$ be one of these triples. The nodes of $m$ can be covered by two paths $(a_{1}^{m}, a_{2}^{m})$ and $(b_{1}^{m}, b_{2}^{m})$. These paths are also feasible since they are arcs from $\setap$.
	
	\begin{figure}
		\begin{center}
			\psscalebox{1.0 1.0} 
			{
		\includegraphics{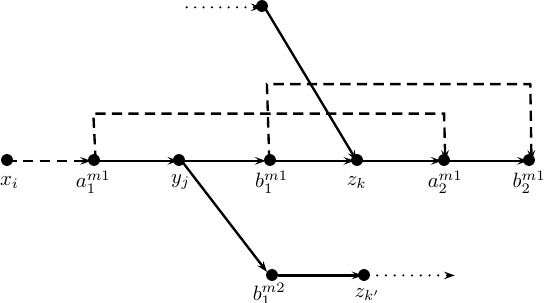}
			}
			\caption{Path of two different gadgets $G_{m_1}$ and $G_{m_2}$}\label{Gmnonf}
		\end{center}
	\end{figure}
	
	Conversely, we assume that \texttt{FPC-ARC} has a solution. Observe that the only way to cover the nodes of $Z$ with a path containing an arc from $\setap$ is by a path, say $P$, starting from a node of $X$, followed by an internal node of the form $a_{1}^{m_1}$, a node from $Y$, an internal node of the form $b_{1}^{m_2}$ and a node from $Z$. Since the paths are disjoint, we are left, after covering all the nodes of $Z$, with at least $4(|M|-q)$ internal nodes. By construction of the gadgets, these nodes cannot be covered by less than $2(|M|-q)$ feasible node-disjoint paths. On the other hand, the remaining arcs of $\setap$ after covering all the nodes of $Z$ are: $|\setap|-(3q+\sum_{i=1}^{q}(|M_i|-1))=3|M|-3q-|M|+q=2(|M|-q)$, where $M_i=\{(x,y,z) \in M \colon x=x_i\}$. Therefore, the path $P$ should also pass by the internal nodes of the form $a_{2}^{m_3}$ and $b_{2}^{m_3}$, otherwise the remaining arcs of $\setap$ cannot cover the remaining internal nodes. 
	
	Let $M'$ be the set of triples corresponding to the paths having the form of $P$ and covering all the nodes of $ X\cup Y \cup Z $. It is clear that $|M'|=q$, and the triples are mutually disjoint. 
	
	We finish the proof by showing that $M'\subseteq M$.
		We proceed by contradiction. Suppose that there exists a triple in $M'$ which is not in $M$, the corresponding path includes necessarily two internal nodes $a_{1}^{m1}$ and $b_{1}^{m2}$ not of the same triple (see Figure~\ref{Gmnonf}). Thus, the only way to cover $a_{2}^{m1}$ is a path containing $z$ of the triple $m1$. This path should not include $b_{2}^{m1}$; otherwise, $b_{1}^{m1}$ cannot be covered by any path containing an arc from $\setap$. Therefore, we have a path starting from a node of $X$ and containing six nodes rather than seven. This requires more than $2(|M|-q)$ feasible node-disjoint paths to cover all the remaining internal nodes, which contradicts the previous conclusions.
		We have, therefore, $M'\subseteq M$, and the theorem follows.
	\end{proof}
	
	Theorem~\ref{FPCnp-hard} implies that \texttt{MPC-ARC} is also $\mathcal{NP}$-hard in the strong sense for arbitrary DAGs. The restricted problem of the crew scheduling problem described in the introduction, which incorporates the rest block constraint, can be seen as a particular case of \texttt{MPC-ARC}. We refer to this problem as \texttt{CS-MPC-ARC}. In \texttt{CS-MPC-ARC}, the underlying DAG is transitive: if pairing $p_j$ can be operated after pairing $p_i$, and pairing $p_k$ can be operated after pairing $p_j$, then $p_k$ can also be operated after $p_i$. Additionally, the arcs in \(\setap\) follow a structured pattern: if an arc connecting pairings $p_i$ and $p_j$ is in \(\setap\), then all arcs connecting any pairing \(p_l\) that precedes (or coincides with) \(p_i\) and any pairing \(p_k\) that succeeds (or coincides with) \(p_j\) are also in \(\setap\). 
	
	Although the complexity of \texttt{CS-MPC-ARC} remains open, we show in the following theorem that \texttt{MPC-ARC} becomes tractable for transitive closures of paths---a specific case of \texttt{CS-MPC-ARC}. However, from the experiments conducted in Section~\ref{sec:compres} on \texttt{CS-MPC-ARC} instances provided by Air France, it seems that the problem is likely not tractable. If this observation is confirmed with a proof, it would imply that adding the rest block constraint to the crew scheduling problem renders it $\mathcal{NP}$-hard, thereby shedding light on the minimal components responsible for the computational complexity of crew scheduling problems.
	
	\begin{theorem}
		\texttt{MPC-ARC} can be solved in polynomial time on DAGs that are transitive closures of paths.
	\end{theorem}
	\begin{proof}
		Let $P$ be a path given by the node sequence $(v_1,\ldots,v_n)$, and $G=(V,A)$ be the transitive closure of $P$ defined by the set of vertices $V=\{v_1,\ldots,v_n\}$ and set of arcs $A=\{(v_i, v_j): 1 \leq i < j \leq n\}$. One can check that $G$ is a DAG. Let $\setap$ be a non-empty subset of $A$. If $\setap$ has at least one arc of the form $(v_i,v_{i+1})$ with $i\in\{1,\ldots,n-1\}$, the path $P$ is the solution of the \texttt{MPC-ARC} problem on $G$. 
		
		\begin{figure}
			\begin{center}
			
			\includegraphics[width=0.85\textwidth]{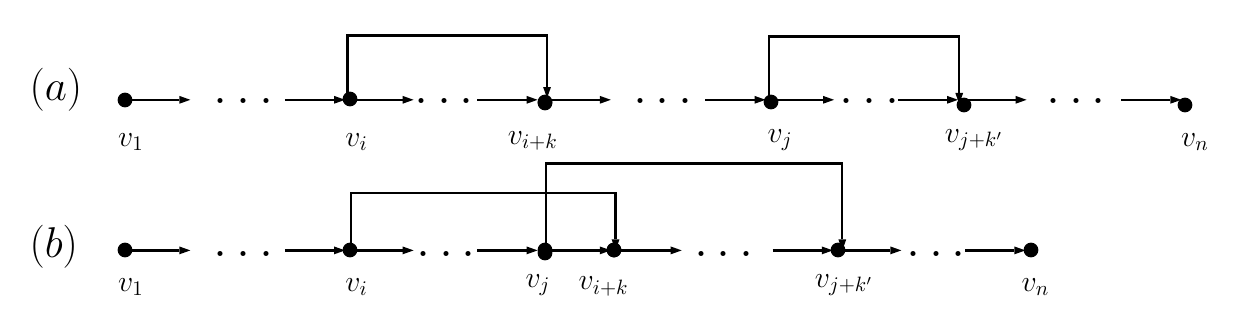}
				
			\end{center}
			\caption{Feasible cases}\label{polyfeasible}
		\end{figure}
		
		\begin{figure}
			\begin{center}
					\includegraphics[width=0.85\textwidth]{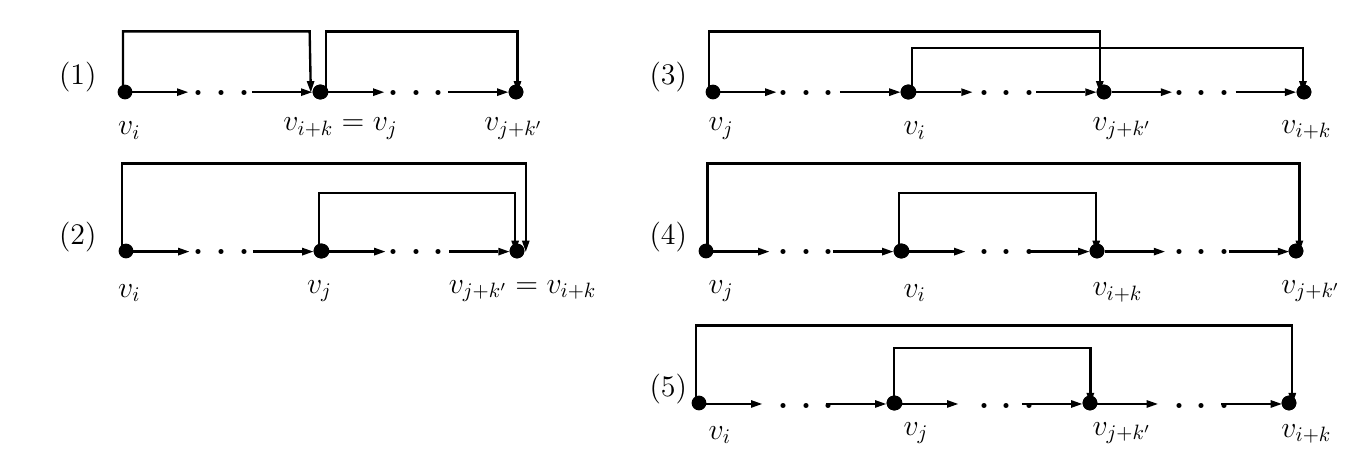}
				
			\end{center}
			\caption{Infeasible cases}\label{polynonfeasible}
		\end{figure}
		
		\sloppy We assume now that the arcs of $\setap$ are of the form $(v_i,v_{i+k})$ with $i=1,\dots,n-2$, $k >1$ and $i+k \leq n$. Observe that, in this case, any optimal solution has at least two paths since any path including an arc $(v_i,v_{i+k}) \in \setap$ cannot cover the nodes $v_{i+1},\ldots, v_{i+k-1}$. If there exist in $\setap$ two arcs of the form $(v_i,v_{i+k})$ and $(v_j,v_{j+k'})$ with $i+k<j$ (resp. with $i+k=j+1$), an optimal solution is formed by the two paths $(v_1,\ldots,v_i,v_{i+k},v_{j+1},\ldots,v_{j+k'-1})$ and $(v_{i+1},\ldots,v_{i+k-1},v_{i+k+1},\ldots,v_{j},v_{j+k'},\ldots,v_n)$ (resp. $(v_1,\ldots,v_i,v_{i+k},\ldots, v_{j+k'-1})$ and $(v_{i+1},\ldots,v_{j},v_{j+k'},\ldots,v_n)$), see Figure~\ref{polyfeasible}.$(a)$ (resp. $(b)$). 
		
		We finish the proof by showing that the remaining cases are not feasible. If $\setap$ contains only one arc $(v_i,v_{i+k})$, the nodes $v_{i+1}, \ldots, v_{i+k-1}$ cannot be covered by any feasible path. We consider now the cases with $|\setap|>2$. Suppose by contradiction that the problem admits a feasible solution $S$ and let $(v_i,v_{i+k})$ be an arc of $\setap$ that is in a path of $S$. To cover the nodes $v_{i+1},\ldots, v_{i+k-1}$, we need a feasible path passing through an arc $(v_j,v_{j+k'})$ of $\setap$. Figure~\ref{polynonfeasible} presents the remaining possible positions of any two arcs of $\setap$. If the two arcs share one endpoint (see cases (1) and (2) of Figure~\ref{polynonfeasible}), the resulting paths are not disjoint. If $\{i+1, i+k-1\}\cap \{j+1,j+k'-1\} \neq \emptyset$ (see cases (3), (4) and (5) of Figure~\ref{polynonfeasible}), there is at least one node in $\{v_{i+1},\ldots, v_{i+k-1}\}$ that cannot be covered by the path that uses $(v_j,v_{j+k'})$. This contradicts the fact that $S$ is feasible since it does not cover all the nodes of $G$.
		
	\end{proof}

	\section{Path-flow based mathematical formulations}\label{sec:mathforms}
	
	In this section, we describe two mathematical formulations for the \fprob\ problem: three-index path flow (\S\ref{sec:mf-tipf}) and two-index path flow (\S\ref{sec:mf-twoipf}) formulations. The formulations are inspired by classical formulations designed for routing problems, which have been proven effective in developing exact branch-and-cut solution approaches for vehicle routing problems (refer to, for example, \cite{letchford_projection_2006}).
	
	For the sake of the formulations, we introduce graph $\bar{G}=(\bar{V},\bar{A})$ derived from graph $G$ by adding dummy nodes $0$ and $\np=n+1$, i.e., $\bar{V}=\{0\} \cup V \cup \{\np\}$. The set of arcs $\bar{A}$ is defined as $\bar{A}=\{(0,i): i \in V\} \cup A \cup \{(i,\np): i \in V\}$. To each path $P=(v_1,v_2,\dots,v_h)$ in $G$ starting at node $v_1$, visiting nodes $\{v_1,v_2,\dots,v_h\} \subseteq V$, and ending at node $v_h$ is associated a path $\bar{P}=(0,v_1,v_2,\dots,v_h,\np)$ in $\bar{G}$ visiting the same set of nodes of $P$ and starting and ending at nodes 0 and $\np$, respectively.
	
	\sloppy Hereafter, the following notation is used. If $(i, j)$ represents an arc, it is \textit{outgoing} from node $i$ and \textit{incoming} to node $j$. The arc $(i, j)$ is incident to both nodes $i$ and $j$, where $i$ is the \textit{start} node, and $j$ is the \textit{end} node of the arc. 

	For any subset $S \subseteq \bar{V}$, the \textit{outgoing} arcs of set $S$ are denoted $\outgoing{S} = \left\{ (i, j) \in \bar{A} : i \in S, j \in \bar{V} \setminus S \right\}$ and similarly $\incoming{S} = \left\{ (i, j) \in \bar{A} : i \in \bar{V} \backslash S, j \in S \right\}$ for \textit{incoming} arcs. If $S = \{i\}$, $\outgoing{i}$ and $\incoming{i}$ are used for simplicity.
	
	\subsection{A three-index path flow formulation}\label{sec:mf-tipf}
	
	Let $c_a$, $a \in \bar{A}$, be a cost associated with arc $a$ defined as follows: 
	$c_a=1$, if $a=(0,i)$; otherwise, $c_a=-|V|$.
	We denote with $K$ the index set $K=\{1,\dots,m\}$, where $m \leq |\setap|$ is an upper bound on the number of paths of an optimal solution to the \fprob. Let $x^k_{a}$, $k \in K$, $a \in \bar{A}$, be a binary variable equal to one if arc $a$ is traversed by path $k$, 0 otherwise. 
	
	The \fprob\ can be formulated as follows:
	\begin{subequations}
		\begin{align}
			(\mFA)\qquad \min & \sum_{k \in K}\sum_{a \in \bar{A}}c_a x^k_a \label{F3:obj}\\
			\text{s.t.}\; & \sum_{a \in \outgoing{0}}x^k_a \leq 1, \quad \forall k \in K \label{F3:k-from-0}\\
			& \sum_{k \in K}\sum_{a \in \incoming{i}}x^k_a \leq 1, \quad \forall i \in V \label{F3:in-V}\\
			& \sum_{a \in \incoming{i}}x^k_a=\sum_{a \in \outgoing{i}}x^k_a, \quad \forall k \in K, i \in V \label{F3:degree-V}\\
			& \sum_{a \in \setap}x^k_a \geq \sum_{a \in \outgoing{0}}x^k_a, \quad \forall k \in K \label{F3:cons-pc}\\
			& x^k_a \in \{0,1\}, \quad \forall a \in \bar{A}.
		\end{align}
	\end{subequations}
	The objective function \eqref{F3:obj} states to maximize the number of covered nodes and to minimize the number of selected paths. Indeed, the cost of a path $\bar{P}=(0,v_1,v_2,\dots,v_h,\np)$ in $\bar{G}$ is equal to $1-h|V|$. The cost of a solution composed of a collection of $q \leq m$ paths $(\bar{P}_1,\bar{P}_2,\dots,\bar{P}_q$) visiting $p_1,p_2,\dots,p_q$ nodes from set $V$, respectively, is equal to $q-(\sum_{s=1}^q p_q)|V|$.
	Constraints \eqref{F3:k-from-0} ensure that each index $k$ is associated with at most one arc leaving node 0. Constraints \eqref{F3:in-V} state that, at most, one arc can enter a node in $V$. Constraints \eqref{F3:degree-V} impose degree constraints for nodes in $V$. Finally,  constraints \eqref{F3:cons-pc} impose that each path in solution contains at least one arc from set $\setap$.
	
	A trivial upper bound on the number of paths of any optimal \prob\ solution can be derived from formulation \mFA\ by first redefining the index set $K=\{1,\dots,|\setap|\}$, and then by solving the following linear program:
	\begin{subequations}
		\begin{align}
			\max & \sum_{k \in K}\sum_{a \in \delta^+(0)} x^k_a \label{F3m:obj}\\
			\text{s.t.}\; & \eqref{F3:k-from-0}-\eqref{F3:cons-pc}\\
			& x^k_a \geq 0, \quad \forall a \in \bar{A}.
		\end{align}
	\end{subequations}
	Let $z$ be the optimal solution cost of the above problem. Then, $m=\lfloor z \rfloor$ represents a valid upper bound.
	
	\subsection{A two-index path flow formulation}\label{sec:mf-twoipf}
	
	Given a path $\bar{P}=(v_0=0,v_1,v_2,\dots,v_h,v_{h+1}=\np)$ in $\bar{G}$, path $\bar{P}$ is 
	\textit{feasible} if it traverses at least one arc in $\setap$, i.e., there exists a pair of nodes $\{v_p,v_{p+1}\} \subseteq \{v_1,v_2,\dots,v_h\}$ such that $(v_p,v_{p+1}) \in \setap$; otherwise, the path is termed  \textit{infeasible}.
	Let $y_a$, $a \in \bar{A}$, be a binary variable equal to one if arc $a$ is used in the solution, zero otherwise. 
	
	The \fprob\ can be formulated as follows:
	\begin{subequations}
		\begin{align}
			(\mFB)\qquad \min & \sum_{a \in \bar{A}}c_a y_a \label{F3R:obj}\\
			\text{s.t.} &  \sum_{a \in \incoming{i}}y_a \leq 1, \quad \forall i \in V \label{F3R:in-V}\\
			& \sum_{a \in \incoming{i}}y_a=\sum_{a \in \outgoing{i}}y_a, \quad \forall i \in V \label{F3R:degree-V}\\
			& \sum_{i=0}^h y_{(v_i,v_{i+1})} \leq h, \quad \forall \text{ infeasible path } \bar{P} \label{F3R:cons-inf}\\
			& y_a \in \{0,1\}, \quad \forall a \in \bar{A}. \label{F3R:yvar}
		\end{align}
	\end{subequations}
	The objective function \eqref{F3R:obj} states to maximize the number of covered nodes and to minimize the number of selected paths. Constraints \eqref{F3R:in-V} are degree constraints imposing that each node $i \in V$ is covered at most once. Constraints \eqref{F3R:degree-V} are flow conservation constraints. The Infeasible Path Constraints (IPCs) \eqref{F3R:cons-inf} forbid infeasible paths to be part of the solution.
	
	\section{Classes of valid inequalities}\label{sec:validineq}
	
	Consider the polytope $\polyt=conv\{y \in \Real^{|\bar{A}|}: y \text{ satisfies } \eqref{F3R:in-V}-\eqref{F3R:yvar}\}$ defined as the convex hull of the characteristic vectors of all feasible paths on graph $\bar{G}$ associated with formulation \mFB. This section investigates classes of valid inequalities for $\polyt$.
	
	Firstly, the following proposition holds.
	\begin{proposition}[Trivial Constraint (TIC)]
		The following inequality is valid for $\polyt$.
		\begin{equation}\label{eq:TIC}        
			\sum_{a \in \setap}y_a \geq \sum_{a \in \outgoing{0}}y_a.
		\end{equation}
	\end{proposition}
	Its validity follows directly from the fact that the number of arcs in the solution belonging to set $\setap$ must be greater than or equal to the number of selected paths in the solution. 
	
	Below, we investigate additional classes of valid inequalities based on valid inequalities designed for related problems to the \fprob.
	
	\subsection{Tournament constraints}
	
	In their work~\cite{ascheuer2000}, the authors proposed various valid inequalities for the Asymmetric Traveling Salesman Problem with Time Windows (\texttt{ATSPTW}), particularly \textit{infeasible path-elimination constraints}. They introduced the basic form of constraints \eqref{F3R:cons-inf} used by formulation \mFB, along with several possibilities for enhancing these inequalities. Among these enhancements are the \textit{tournament constraints}, which provide valid inequalities for $\polyt$.
	
	The tournament constraints strengthen the IPCs ~\eqref{F3R:cons-inf} as follows. By abuse of notation, we assume that variables $y_{(v_i,v_j)}$ are indexed over all the pairs $v_i, v_j \in \bar{V}$, $v_i \neq v_j$, but they appear in the corresponding expressions only for node pairs $(v_i,v_j) \in \bar{A}$.
	
	Consider an infeasible path $\bar{P} = (v_0=0, v_1, v_2, \ldots, v_h, v_{h+1}=\np)$ in $\bar{G}$. 
	\begin{proposition}[Tournament Constraints I (TC-I)]
		The following inequality is valid for $\polyt$.
		\begin{equation}\label{tournament} \sum_{i=0}^{h}\sum_{j=i+1}^{h+1} y_{(v_i,v_j)} \leq h. \end{equation} 
	\end{proposition} 
	\begin{proof}
		We observe that since $\delta^+(v_i)\leq 1$, $i=0,\dots,h$, in any optimal \fprob\ solution, we have $\sum_{j=i+1}^{h+1} y_{(v_i,v_j)}\leq 1$, $i=0,\dots,h$. Hence, if $\sum_{i=0}^{h}\sum_{j=i+1}^{h+1} y_{(v_i,v_j)} > h$, then $\sum_{j=i+1}^{h+1} y_{(v_i,v_j)}=1$, $i=0,\dots,h$ and we must have $y_{(v_i,v_{i+1})}=1$, $i=0,\dots,h$, which is a contradiction due to infeasibility of path $\bar{P}$.
	\end{proof}
	
	The authors introduced in~\cite{ascheuer2000} several lifting procedures. Most of these procedures rely on the existence of bidirectional paths between nodes, a condition not met in the context of the \fprob\ problem due to its DAG. However, the $V$-lifting procedure proposed in \cite{ascheuer2000} can be adapted to the \fprob\ problem as follows.
	
	\begin{proposition}[Tournament Constraints II (TC-II)]
		Consider an infeasible path $\bar{P} = (v_0=0, v_1, v_2, \ldots, v_h, v_{h+1}=\np)$ and a node $v_k \notin \{0,v_1,\ldots,v_h,\np\}$ such that path $\bar{P}'=(0,v_1,\ldots,v_l,v_k,v_{l+1},\ldots,v_h,\np)$, with $v_{l} \neq 0$ and $v_{l+1} \neq \np$, is infeasible for a certain $v_l$. Then
		\begin{equation}\label{vlifting}
			\bigg ( \sum_{i=0}^{h}\sum_{j=i+1}^{h+1} y_{(v_i,v_j)} \bigg )+ y_{(v_l,v_k)}+y_{(v_k,v_{l+1})}+y_{(v_l,v_{l+1})} \leq h+1
		\end{equation}
		is valid for $\polyt$. 
	\end{proposition}
	\begin{proof}
		The inequality is derived by applying the $V$-lifting procedure described in ~\cite{ascheuer2000} on the infeasible paths $\bar{P}$ and $\bar{P}'$. More specifically, given nodes $v_l$, $v_k$ and $v_{l+1}$ such that $v_{l} \neq 0$ and $v_{l+1} \neq \np$, we first observe that the following inequality is satisfied by any optimal \fprob\ solution:
		\begin{equation}\label{tcii-a}
			y_{(v_l,v_k)} + y_{(v_k,v_{l+1})} + 2y_{(v_l,v_{l+1})} \leq 2.
		\end{equation}
		Indeed, if node $v_l$ is visited in a solution, at most once of the paths $(v_l,v_k,v_{l+1})$ and $(v_l,v_{l+1})$ can be part of the solution.
		We derive a new inequality by a linear combination of the TC-I inequality \eqref{tournament} defined on path $\bar{P}$, the TC-I inequality \eqref{tournament} defined on path $\bar{P}'$ and inequality \eqref{tcii-a} weighting the inequalities by scalar $\frac{1}{2}$.
		We obtain:
		\begin{equation}\label{tcii-b}
			\begin{split}
				\bigg (\sum_{i=0}^h\sum_{j=i+1}^{h+1} y_{(v_i,v_j)}\bigg ) + \bigg (\frac{1}{2}\sum_{i=0}^{l-1}y_{(v_i,v_k)} \bigg ) + y_{(v_l,v_k)}+y_{(v_k,v_{l+1})}+y_{(v_l,v_{l+1})}+ \\ \bigg (\frac{1}{2}\sum_{i=l+2}^{h+1}y_{(v_k,v_i)}\bigg ) \leq h + \frac{3}{2}.
			\end{split}
		\end{equation}
		The left-hand side is integer-valued by rounding down the terms at the left-hand side with coefficients $\frac{1}{2}$. Hence, the right-hand-side coefficient $h + \frac{3}{2}$ can be rounded down to the integer coefficient $h + \lfloor \frac{3}{2} \rfloor = h+1$, thus deriving inequality \eqref{vlifting}.
	\end{proof}
	
	\subsection{$\setap$-reachability constraints}
	For each node $i$ in $\overline{V}$, we define the \textit{$\setap$-reaching arc set} $\setap^{-}_{i}$ as a subset of $\setap$ such that there exists a path from each arc of $\setap^{-}_{i}$ to $i$. Similarly, we define the \textit{$\setap$-reachable arc set} $\setap^{+}_{i}$ as a subset of $\setap$ such that there exists a path from $i$ to each arc of $\setap^{+}_{i}$.
	
	\begin{proposition}
		Let $(i,j) \in \bar{A} \setminus \setap$, the following inequality is valid for $\polyt$.
		\begin{equation}\label{eq:ARC-old}
			\sum_{a\in \setap^{-}_{i}} y_a+  \sum_{a\in \setap^{+}_{j}} y_a \geq y_{(i,j)}.
		\end{equation} 
	\end{proposition}
	\begin{proof}
		By definition of a feasible path and the sets $ \setap^{-}_{i}$ and $ \setap^{+}_{j} $, any path covering $(i,j)$ must include at least one arc from $\setap^{-}_{i} \cup \setap^{+}_{j}$. Therefore, if $(i,j)$ is selected, at least one arc from $\setap^{-}_{i} \cup \setap^{+}_{j}$ must also be selected.
	\end{proof}
	
	Inequalities \eqref{eq:ARC-old} can be strengthened, as stated by the following proposition. To simplify notations, we introduce the auxiliary variables $z_i=\sum_{a \in \delta^-(i)}y_a$, $\forall i \in V$, which equals $1$ if node $i$ is covered and $0$ otherwise.
	
	\begin{proposition}[$\setap$-Reachability Constraints (A-RC)]
		Let $i\in V$, the following inequality is valid for $\polyt$.
		\begin{equation}\label{eq:ARC}
			\sum_{a\in \setap^{-}_{i}} y_a+  \sum_{a\in \setap^{+}_{i}} y_a \geq z_i.
		\end{equation} 
	\end{proposition}
	\begin{proof}
		By definition of a feasible path and the sets $ \setap^{-}_{i}$ and $ \setap^{+}_{i} $, any path covering $i$ must include at least one arc from $\setap^{-}_{i} \cup \setap^{+}_{i}$. Therefore, if $i$ is selected, at least one arc from $\setap^{-}_{i} \cup \setap^{+}_{i}$ must also be selected.
	\end{proof}
	
	Two nodes $i$ and $j$ are \textit{conflicting} if and only if there does not exist a path in $\bar{G}$ that passes through the two nodes. If nodes $i$ and $j$ are conflicting nodes and part of a solution, the two nodes are covered by node-disjoint paths. Let $T \subseteq V$ be a set of conflicting nodes, and $\setap^{{-}}_{T}=\bigcup\limits_{i \in T} \setap^{-}_{i}$ and $\setap^{{+}}_{T}=\bigcup\limits_{i \in T}\setap^{+}_{i}$ be the sets of all $\setap$-reaching and $\setap$-reachable arcs of the nodes in $T$, respectively.
	
	\begin{proposition}[Generalized $\setap$-Reachability Constraints (A-GRC)]
		The following inequality is valid for $\polyt$.
		\begin{equation}\label{eq:A-GRC}
			\sum_{a\in \setap^{-}_{T}} y_a + \sum_{a\in \setap^{+}_{T}} y_a \geq \sum_{i\in T}  z_i.
		\end{equation}
	\end{proposition}
	\begin{proof}
		Since the nodes in $T$ can be covered only by node-disjoint paths, the number of selected arcs from $\setap^{-}_{T}$ and $\setap^{+}_{T}$ must be at least equal to the number of covered nodes from $T$.
	\end{proof}
	
	\subsection{Reachability cuts}
	The simple $(\pi,\sigma)$-inequality, introduced in \cite{balas1995}, defines constraints for ordered pairs of nodes $i$ and $j$ in the precedence-constrained asymmetric traveling salesman problem. These constraints specify that a path from $i$ to $j$ can only traverse a specific subset of arcs determined by precedence constraints. In \cite{ascheuer2000}, the authors proposed an enhancement of the $(\pi,\sigma)$-inequality tailored for the \texttt{ATSPTW}, where time-window restrictions further constrain the feasible arcs.
	Building upon this idea, in \cite{lysgaard2006}, the authors extended the concept to the vehicle routing problem with time windows (\texttt{VRPTW}) to derive a new class of valid inequalities called reachability cuts or R-cuts. Unlike \texttt{ATSPTW}, \texttt{VRPTW} does not necessarily require visiting all nodes on a single route. The authors in \cite{lysgaard2006} recognized that since all paths pass through the depot, it is feasible to mandate paths in both directions between the depot and any other node. This insight led to the adaptation of these cuts to the \fprob\ problem, as described in the following.
	
	For each $i \in V$, let $A^{-}_{i}$ (resp. $A^{+}_{i}$) be the set of arcs of all feasible paths of the form $(0,\ldots,i)$ (resp. $(i,\ldots,\np)$). We refer to these sets as the \textit{reaching arc set} for $A^{-}_{i}$ and the \textit{reachable arc set} for $A^{+}_{i}$. For any subset $T \subseteq V$, we define $A^{-}_{T}=\bigcup_{i \in T} A^{-}_{i}$, similarly $A^{+}_{T}=\bigcup_{i \in T} A^{+}_{i}$. 
	
	\begin{proposition}[Reaching Constraints (RC$\pm$)]\label{propRC}
		For any set $S \subseteq V$ and any subset $T \subseteq S$ of conflicting nodes, the following inequality is valid for $\polyt$.
		\begin{equation}\label{rcuts}
			\sum_{a \in \incoming{S} \cap A^{-}_{T}} y_a +  \sum_{a \in \outgoing{S} \cap A^{+}_{T}} y_a \geq \sum_{i\in T}  z_i.
		\end{equation}
	\end{proposition}
	\begin{proof}
		The selected nodes from $T$ must be reached from node $0$ or reach node $\np$ via node-disjoint feasible paths. Consequently, the number of selected arcs from $A^{-}_{T}$ crossing the cutset $(\bar{V} \setminus S : S)$ plus the number of selected arcs from  $A^{+}_{T}$ crossing the cutset $(S:\bar{V} \setminus S)$ must be at least equal to the number of selected nodes from $T$.
	\end{proof}
	
	For each $i \in V$, let $\grave{A}^{-}_{i}$ (resp. $\grave{A}^{+}_{i}$) be the set of arcs of all partial paths of the form $(0,\ldots,i)$ (resp. $(i,\ldots,\np)$) that are part of a feasible path from $0$ to $\np$. Based on the RC$\pm$, the following inequalities are also valid, whose proofs rely on a similar argument of the proof of Proposition \ref{propRC}.
	
	\begin{proposition}[Reaching Constraints (RC-)]
		For any set $S \subseteq V$ and any subset $T \subseteq S$ of conflicting nodes, the following inequality is valid for $\polyt$.
		\begin{equation}\label{-rcuts-new}
			\sum_{a \in \incoming{S} \cap \grave{A}^{-}_{T}} y_a \geq \sum_{i\in T}  z_i.
		\end{equation}
	\end{proposition}
	
	\begin{proposition}[Reachable Constraints (RC+)]
		For any set $S \subseteq V$ and any subset $T \subseteq S$ of conflicting nodes, the following inequality is valid for $\polyt$.
		\begin{equation}\label{+rcuts-new}
			\sum_{a \in \outgoing{S} \cap \grave{A}^{+}_{T}} y_a \geq \sum_{i\in T}  z_i.
		\end{equation}
	\end{proposition}
	
	\subsection{Generalized cut constraints}
	
	\cite{fischetti2001} introduced the generalized cut constraints (GCUTs) for the 
	\texttt{VRPTW}. In GCUTs, a binary indicator $\gamma_{ij}$ is defined for each pair of nodes $i$ and $j$. This indicator is set to $1$ if a feasible path exists in the form $(0, \ldots, i, \ldots, j, \ldots, \np)$, and $0$ otherwise. For a given node set $S$, let $\Pi(S) = \{i \in V\setminus S : \gamma_{ij}=0 \text{ for all } j \in S\}$.
	\begin{proposition}[Generalized cut constraints (GCUTs)]
		For a given set $S \subset V$ and a conflicting set $T \subset S$, the following inequality is valid for $\polyt$.
		\begin{equation} \label{GCUT}
			\sum_{i \in (V \setminus S)\setminus \Pi(T), \; j \in S \setminus \Pi(T)}
			y_{(i,j)} \geq  \sum_{i\in T}  z_i.
		\end{equation}
	\end{proposition}
	One can observe that \eqref{GCUT} is dominated by \eqref{rcuts}. Each arc on the left-hand side of \eqref{rcuts} also appears on the left-hand side of \eqref{GCUT}. Specifically, for every arc $a=(i,j)$ where $(i,j) \in A^{-}_{i}$, we have $i$ and $j$ not in $\Pi(T)$. However, there might be arcs in the left-hand side of \eqref{GCUT} that are not considered in \eqref{rcuts}. For instance, if $T=\{k\}$, $i \notin S$, and $j \in S$, feasible paths of the form $(0,\ldots,i,k)$ and $(0,\ldots,j,k)$ could exist, but there might not be a feasible path to $k$ that includes $(i,j)$, especially if $\setap=\{(i,k),(j,k)\}$. Consequently, $y_{(i,j)}$ appears in the left-hand side of \eqref{GCUT} but not in \eqref{rcuts}.
	
	\section{Exact branch-and-cut algorithm}\label{sec:exactm}
	This section describes a branch-and-cut (BC) algorithm for solving the \fprob\ based on formulation \mFB\ strengthened with the different valid inequalities described in the previous section.
	
	The algorithm was developed within the Gurobi framework using Gurobi callback functions. These functions allow the programmer to customize the general approach embedded in Gurobi extensively. For instance, one can select the next node to explore in the enumeration tree, choose the branching variable, define a problem-dependent branching scheme, incorporate their cutting planes, apply custom heuristic methods, and more. For additional details about the usage of these callback functions, readers can refer to the documentation of the Gurobi callable library \citep{gurobi}. Our work focused on designing and implementing separation algorithms for the various cuts and leveraging \citep{gurobi} callback functions to handle user constraints.
	
	\subsection{Separation routines}
	
	In this section, we describe the separation algorithms for the classes of inequalities listed in Section \ref{sec:validineq}.
	
	Given a solution $y^\ast \in \Real^{|\bar{A}|}$ of the Linear Programming (LP) relaxation of formulation \mFB, let $G_{y^\ast} = (\bar{V},\{(i,j) \in \bar{A} : y^\ast_{ij} > 0\})$ be the \textit{support graph} associated with solution $y^\ast$, and let $z^\ast_i = \sum_{a \in \delta^+(i)} y^\ast_a$, for all $i \in V$. Clearly, graph $G_{y^\ast}$ is a DAG.
	
	We first observe that the TIC \eqref{eq:TIC} can be easily separated in polynomial time. 
	
	\subsubsection*{IPCs \eqref{F3R:cons-inf}}
	
	The separation problem associated with the IPCs asks for finding a path $\bar{P}^\ast = (v_0=0,v_1,v_2,\dots,v_h,v_{h+1}=\np)$ in $G_{y^\ast} \setminus \setap$ such that $\bar{P}^\ast = \argmax_{\bar{P} \text{ in } G_{y^\ast}} \big \{\sum_{i=0}^h y^\ast_{(v_i,v_{i+1})} - h\big \}$. Since $G_{y^\ast}$ is acyclic, the path $\bar{P}^\ast$ can be computed in polynomial time by finding the longest path in the graph $G_{y^\ast}$ from node 0 to node $\np$, using the cost matrix $[u_{ij}]$, where the cost of arc $(i,j)$ is computed as $u_{ij} = y^\ast_{ij}-1$. Let $\alpha$ be the optimal solution cost of the longest path (if any). If $\alpha +1 > 0$, then a violated IPC is found, and path $\bar{P}^\ast$ corresponds to the path maximizing the violation of the constraints.
	
	\subsubsection*{TC-I \eqref{tournament} and TC-II \eqref{vlifting}}
	
	To separate the TC-I, we employ the same separation procedure used for the IPCs to identify any violated IPC. The resulting inequality (if found) is then lifted to derive the corresponding TC-I. TC-II inequalities are also derived from a violated IPC by iteratively checking for the existence of a node $v_k$ not belonging to the IPC. If adding this node to the path results in another IPC violation, a TC-II inequality is generated for each violated IPC found during these iterations, along with the corresponding TC-I inequality.
	
	\subsubsection*{A-RC \eqref{eq:ARC} and A-GRC \eqref{eq:A-GRC}}
	
	We first need to determine the $\setap$-reaching and $\setap$-reachable arc sets for each $i \in V$. We perform a Depth-First Search (DFS) in $\bar{G}$. For each arc in $\setap$, we check whether a path exists from this arc to a node $i$; if such a path exists, we add the arc to $\setap^{-}_{i}$. Similarly, for each arc in $\setap$, we check whether a path exists from a node $i$ to this arc; if such a path exists, we add the arc to $\setap^{+}_{i}$.
	
	The separation problem associated with the A-GRCs~\eqref{eq:A-GRC} consists of finding a set $T^\ast$ of conflicting nodes in $\bar{G}$ such that  $T^\ast \in
	\argmax_{T \subset V}\left \{\sum_{i\in T}  z_i^\ast -  \sum_{a\in \setap^{-}_{T} \cup \setap^{+}_{T}} y_a^\ast \right \}$.  Let $\tilde{G}_w = (\bar{V}, \tilde{E})$ be a weighted graph constructed from $\bar{G}$ as follows. The node set remains $\bar{V}$. The edge set $\tilde{E}$ includes the edge $\{i,j\}$ if and only if a path exists from $i$ to $j$ or from $j$ to $i$ in $\bar{G}$. The weight of each node $i \in \bar{V}$ is given by $W(i) = z^\ast_i - \sum_{a \in \setap^{-}_{i} \cup \setap^{+}_{i}} y^\ast_{a}$. Observe that a set of conflicting nodes in  $\bar{G}$ corresponds to an independent set in $\tilde{G}_w $. Therefore, the exact separation of the A-GRCs~\eqref{eq:A-GRC} requires finding an independent set $T^\ast$ in $\tilde{G}_w$ that maximizes $\sum_{i \in T^\ast} z^\ast_i - \sum_{a \in \setap^{-}_{T^\ast} \cup \setap^{+}_{T^\ast}} y^\ast_{a}$.
	
	We adapt a greedy heuristic, namely \texttt{GWMIN}, proposed in~\cite{SMK03} for the maximum weighted independent set problem. The \texttt{GWMIN} heuristic selects a node $i$ that maximizes $W(i)/(d_{\tilde{G}_w^h}(i) + 1)$, where $d_{\tilde{G}_w^h}(i)$ is the degree of $i$ in $\tilde{G}_w$ during the $h^{\text{th}}$ iteration. After selecting node $i$, it and its neighbors are removed from the graph. This process repeats until no nodes remain, and the selected nodes form the independent set to be returned.
	
	We modify this heuristic to separate the A-GRCs~\eqref{eq:A-GRC} to account for the arcs of $\setap$ that are shared between the selected nodes. Specifically, at each iteration, when a node $i$ is selected, we update the weights of the remaining nodes \(j\) with \((\setap^{-}_{i} \cup \setap^{+}_{i}) \cap (\setap^{-}_{j} \cup \setap^{+}_{j}) \neq \emptyset\). The updated weight for node \(j\) is calculated as \(z^\ast_j - \sum_{a \in (\setap^{-}_{j} \cup \setap^{+}_{j}) \setminus (\setap^{-}_{i} \cup \setap^{+}_{i})} y^\ast_{a}\). Let \(W^\ast\) be the weight of the independent set found. If \(W^\ast > 0\), then a violated A-GRC is identified.
	
	In the implementations, we also explored an alternative approach using the exact algorithm provided in Cliquer~\cite{niskanen2003cliquer} to find a maximum-weight independent set in $\tilde{G}_w$.
	
	\subsubsection*{RC$\pm$ \eqref{rcuts}}\label{sep:rcuts}
	
	The identification of $A^{-}_{i}$ and $A^{+}_{i}$ for each $i \in V$ is also based on DFS. After running this algorithm, we examine each arc $(j,k) \in \bar{A}$ to check whether there is a path from $k$ to $i$ and whether there exists an arc $(u,v) \in \setap$ such that there is a path from $v$ to $j$, or a path from $k$ to $u$ and from $v$ to $i$. If these conditions are met, the arc $(j,k)$ is included in $A^{-}_{i}$. The set $A^{+}_{i}$ is computed similarly.
	
	For a given set of conflicting nodes $T$, the separation of RC$\pm$ \eqref{rcuts} consists of finding a set $S^\ast \subseteq \bar{V}$ such that $S^\ast \in \argmin_{S \subseteq \bar{V} s.t. S \supseteq T}\left \{\sum_{a \in \incoming{S} \cap A^{-}_{T}} y_a^\ast  +  \sum_{a \in \outgoing{S} \cap A^{+}_{T}} y_a^\ast \right \}$. We show in the following that finding a set $S^*$ reduces to solving two maximum flow problems derived from the two components of the latter expression. 
	
	Let $\tilde{G}_1$ be the graph with node set $\bar{V}$ and arc set $A^{-}_{T} \cup \{(i,\np) : i \in T\}$. Each arc $(i,j) \in A^{-}_{T}$ is assigned a capacity $y^\ast_{ij}$, while each arc in $\{(i,\np) : i \in T\}$ has an infinite capacity. Similarly, let $\tilde{G}_2$ be the graph with node set $\bar{V}$ and arc set $A^{+}_{T} \cup \{(0,i) : i \in T\}$. Each arc $(i,j) \in A^{+}_{T}$ is assigned a capacity $y^\ast_{ij}$, and each arc in $\{(0,i) : i \in T\}$ has an infinite capacity.
	
	Given a minimum cut in $\tilde{G}_1$, define $S_1^\ast$ as the set of nodes on the sink side of the cut, excluding both isolated nodes in $\tilde{G}_1$ and the sink itself. Clearly, $S_1^\ast$ is an element of the set  $\argmin_{S \subseteq \bar{V} s.t. S \supseteq T} \left\{\sum_{a \in \incoming{S} \cap A^{-}_{T}} y^\ast_a \right\}$. Similarly, for a minimum cut in $\tilde{G}_2$, define $S_2^\ast$ as the set of nodes on the source side of the cut, excluding both isolated nodes in $\tilde{G}2$ and the source. Clearly, $S_2^\ast$ is an element of the set $\argmin_{S \subseteq \bar{V} s.t. S \supseteq T} \left \{ \sum_{a \in \outgoing{S} \cap A^{+}_{T}} y^\ast_a \right \}$. The following proposition holds.

	\begin{proposition}
		$S_1^\ast \cup S_2^\ast \in \argmin_{S \subseteq \bar{V} s.t. S \supseteq T} \left\{\sum_{a \in \incoming{S} \cap A^{-}_{T}} y^\ast_a \right\}$ and $S_1^\ast \cup S_2^\ast \in \argmin_{S \subseteq \bar{V} s.t. S \supseteq T} \left \{ \sum_{a \in \outgoing{S} \cap A^{+}_{T}} y^\ast_a \right \}$.
	\end{proposition}
	
	\begin{proof}
		First, observe that no node $i \in \bar{V} \setminus T$ can be incident to arcs in both $A^{-}_{T}$ and $A^{+}_{T}$ simultaneously. If this were the case, there would be a path from some node in $T$ to $i$ and another path from $i$ to a node in $T$. This would either contradict the fact that the nodes in $T$ are conflicting (if the two nodes in $T$ are different) or violate the acyclic nature of the graph (if the two nodes in $T$ are the same). On the other hand, by construction of $S_1^\ast$ and $S_2^\ast$, all nodes in $S_1^\ast \setminus T$ are incident to arcs in $A^{-}_{T}$, and all nodes in $S_2^\ast \setminus T$ are incident to arcs in $A^{+}_{T}$. 
		
		Notice that $(S_1^\ast \cup S_2^\ast) \setminus S_1^\ast = S_2^\ast \setminus T$, so all nodes in $(S_1^\ast \cup S_2^\ast) \setminus S_1^\ast$ cannot be incident to any arcs in $A^{-}_{T}$. Consequently, $\sum_{a \in \incoming{S_1^\ast \cup S_2^\ast} \cap A^{-}_{T}} y^\ast_a =  \sum_{a \in \incoming{S_1^\ast} \cap A^{-}_{T}} y^\ast_a $. Similarly, $(S_1^\ast \cup S_2^\ast) \setminus S_2^\ast = S_1^\ast \setminus T$, so all nodes in $(S_1^\ast \cup S_2^\ast) \setminus S_2^\ast$ cannot be incident to any arcs in $A^{+}_{T}$. Therefore, $\sum_{a \in \incoming{S_1^\ast \cup S_2^\ast} \cap A^{+}_{T}} y^\ast_a =  \sum_{a \in \incoming{S_1^\ast} \cap A^{+}_{T}} y^\ast_a $.
	\end{proof}
	
	From this proposition, we can conclude that $S_1^\ast \cup S_2^\ast \in \argmin_{S \subseteq \bar{V} s.t. S \supseteq T}\left \{\sum_{a \in \incoming{S} \cap A^{-}_{T}} y_a^\ast  +  \sum_{a \in \outgoing{S} \cap A^{+}_{T}} y_a^\ast \right \}$. Thus,  $S^\ast$ can be determined by solving two maximum flow problems in $\tilde{G}_1$ and $\tilde{G}_2$.
	
	For the computation of $T$, we use in the experiments the greedy heuristic GWMIN~\cite{SMK03} and the exact algorithm from Cliquer~\cite{niskanen2003cliquer}, and compared the results obtained in both cases.
	
	\subsubsection*{RC- \eqref{-rcuts-new} and RC+\eqref{+rcuts-new}}
	
	The identification of $\grave{A}^{-}_{i}$ and  $\grave{A}^{+}_{i}$ for each $i \in \bar{V}$ is also based on DFS. To determ$\grave{A}^{-}_{i}$, we first check if there is a path from $i$ to an arc in $\setap$. If such a path exists, then $\grave{A}^{-}_{i}$ includes all arcs $(i,k) \in \bar{A}$ for which a path from $k$ to $i$ exists. Otherwise, $\grave{A}^{-}_{i}=A^{-}_{i}$, where $A^{-}_{i}$ is determined as described in Section~\ref{sep:rcuts}. The set $\grave{A}^{+}_{i}$ is identified similarly.
	
	
	We employ the separation procedure proposed in~\cite{lysgaard2006} to separate RC- \eqref{-rcuts-new} and RC+\eqref{+rcuts-new}. For a given set of conflicting nodes $T$, the separation problem for RC- \eqref{-rcuts-new} involves finding a set $S^\ast \subseteq \bar{V}$ such that $S^\ast \in \argmin_{S \subseteq \bar{V} s.t. S \supseteq T}\left \{\sum_{a \in \incoming{S} \cap \grave{A}^{-}_{T}} y_a^\ast \right \}$.
	
	To find this set, we construct a graph with the node set $\bar{V}$ and the arc set $\grave{A}^{-}_{T} \cup \{(i,\np): i \in T\}$. Each arc $(i,j) \in \grave{A}^{-}_{T}$ is assigned a capacity of $y^\ast_{ij}$, while each arc in $\{(i,\np): i \in T\}$ is given infinite capacity. Finding $S^\ast$ corresponds to finding a minimum cut in this graph. After solving the maximum flow problem with $0$ as the source and $\np$ as the sink, the nodes on the sink side of the minimum cut constitute the set $S^\ast$. If the maximum flow value is less than $\sum_{i\in T}  z^\ast_i$, then a violated RC- is identified. The separation of RC+\eqref{+rcuts-new} is performed similarly.
	
	In the experiments, we compute set $T$ using the greedy heuristic GWMIN~\cite{SMK03} and the exact algorithm from Cliquer~\cite{niskanen2003cliquer}, similar to the previous classes of cuts, and compare the results obtained with each approach.
	
	\subsection{Separation, branching, and node selection strategies}
	
	When the solution to the LP relaxation at a given node is integral, Gurobi invokes the \codef{Callback} method to assess its feasibility. If an IPC or a TIC is identified, the constraints are added to the master problem using the \codef{addLazy} function.
	
	When an LP relaxation is solved and the node is not pruned, the current master solution is examined to identify any violated inequalities. Based on preliminary computational experiments, the inequalities are separated in the following order according to their computational complexity: (1) IPC, TC-I, and TC-II; (2) A-GRC; (3) RC$\pm$, RC-, and RC+. If any of these procedures generate a cutting plane, the subsequent routines are bypassed. For all three variants, IPC and TIC are included as lazy cuts. Note that TIC, if added, is included only once. When violated inequalities are identified, they are incorporated into the master problem using the \codef{addCut} function.
	
	Following the results of preliminary experiments, we have adopted the ``lowest first'' strategy as our node-selection rule.
	
	\section{Computational experiments}\label{sec:compres}
	
	This section presents the results of the BC algorithm described in Section \ref{sec:exactm}. Additionally, the Appendix \ref{sec:appendix} presents detailed results from the runs of various BC algorithm variants, along with the results of formulation $\mFA$. 
	All the implementations were done in C++ on a machine equipped with an 11\textsuperscript{th} Gen Intel\textsuperscript{\textregistered} Core\textsuperscript{TM} i5-11500 @ 2.70GHz~$\times$~12 processor and 32 GiB of RAM. We used the Gurobi 11.0.2 solver \citep{gurobi}. 
	
	In Section~\ref{sec:instances}, we describe the test instances, and in Section~\ref{sec:results}, we present the results obtained for these sets,
	including the effectiveness of root node cuts and details about the best-performing BC variant for each set.
	
	\subsection{Instances} \label{sec:instances}
	
	We considered three sets of instances: A, B, and C. Set A, and C were randomly generated, while set B was derived from real-world instances of the crew scheduling problem at Air France, as used in~\cite{tellache2024linear}.
	
	The instances of A are generated as follows. We first construct a complete DAG (transitive closure of a path) with $n$ nodes. Then, we remove arcs randomly from the graph with probability $(1-p_{a})$, provided that the resulting graph remains connected. The arcs of $\setap$ are subsequently selected from $A$ randomly with probability $p_{a'}$. All random selections in this process follow a uniform distribution.
	
	When running the exact BC algorithm on this set of instances, we observed that it consistently produces optimal solutions, particularly when $p_{a'} \geq 0.2$. We attribute this to enough arcs in $\setap$ that are uniformly distributed in the DAG, allowing the optimal MPC solution to be easily transformed into a feasible solution under $\setap$-related constraints by adding only a few lazy cuts in most cases. Based on this observation, we divided these instances into two groups: group A.1, containing instances with $p_{a'} \geq 0.2$, and group A.2 with $p_{a'} < 0.2$. In the second group, the number of arcs in $\setap$ is between $3$ and $19$, making it difficult to guarantee the existence of an MPC. For each combination of $(n,p_{a},p_{a'})$ in A.1, five instances were generated. In contrast, for group A.2, 150 instances were generated, from which the five most difficult instances were selected based on the variant of the BC algorithm that uses only lazy constraints of type IPCs and TIC. In set A.1, we considered all the combinations of $n \in \{100, 300, 500\}$ and $p_{a}, p_{a'} \in \{0.2, 0.5, 0.8\}$, yielding a total of 135 instances. In set A.2, we also considered $n \in \{100, 300, 500\}$ and $p_{a}\in \{0.2, 0.5, 0.8\}$, while $p_{a'}$ varied from $0.01$ to $0.0001$, giving a total of 45 instances across nine groups. The characteristics of sets A.1 and A.2 are summarized in Tables~\ref{tab:instA1C} and~\ref{tab:instA2}, respectively, where $\#inst$ denotes the number of instances in each group.
	
	\begin{table}
		\caption{Characteristics of instance sets A.1 and C} \label{tab:instA1C}
		\begin{center}
			{\scalebox{1.0}{
					\begin{tabular}{rrrrrrrr}
						\toprule
						&&&&\multicolumn{2}{c}{ A.1} & \multicolumn{2}{c}{ C}  \\
						\cmidrule{5-6}\cmidrule{7-8}
						$n$&$p_a$&$p_{a'}$  & $\#inst$  & Average number  & Average number& Average number  & Average number \\ 
						&& &   & of arcs &  of arcs in $\setap$& of arcs &  of arcs in $\setap$ \\
						\midrule
						100	&0.2	&0.2&	5& 975.8&   201.0  &  976.6&      195.0\\
						&&0.5&	5&1,000.8&     493.0&  969.2&    484.2\\
						&&0.8&	5& 982.0&        788.4&  976.0&     780.4\\
						&0.5&	0.2	&5& 2,478.4&  493.6& 2,497.8&     499.0\\
						&&0.5	&5& 2,463.2&   1,235.4 &  2,473.0&    1,236.2\\
						&&0.8	&5&  2,463.8&  1,958.6&  2,584.2&    2,067.0\\
						&0.8&	0.2	&5& 3,960.2&  770.8&  3,971.0&   793.8\\
						&&0.5&	5& 3,947.6&   1,974.2&  4,008.8&    2,004.0\\
						&&0.8&	5& 3,941.0&    3,147.2&  3,980.8&    3,184.4\\
						\midrule
						300	&0.2&	0.2&	5& 8,906.0&   1,793.8& 8,934.8&   1,786.2\\
						&&0.5&5&	 9,010.8&   4,470.6&  8,910.4&   4,454.6\\
						&&0.8	&5& 8,973.0&  7,200.4&  8,931.4&    7,144.2\\
						&0.5	&0.2	&5& 22,414.0&  4,511.0&   22,347.4&   4,468.8\\
						&&0.5&	5& 22,441.8&    11,222.6& 22,356.6&    11,177.2\\
						&&0.8&	5&22,432.4&       17,932.6&  22,385.4&   17,907.2\\
						&0.8	&0.2&5&	 35,864.4&   7,237.4&  35,738.8&    7,146.8\\
						&&0.5&	5& 35,905.8&    17,980.0& 35,780.8&    17,889.8\\
						&&0.8&	5&35,897.6&   28,738.8&35,759.6&     28,607.0\\
						\midrule
						500	&0.2	&0.2	&5&  25,024.2&    4,976.4&   24,897.4&    4,978.8\\
						&&0.5&	5&  25,055.6&   12,538.6&    24,841.0&   12,419.8\\
						&&0.8&	5&24,948.4&    19,950.2&  24,920.2&      19,935.4\\
						&0.5&	0.2&	5&   62,469.4&     12,541.8& 62,289.8&   12,456.4\\
						&&0.5	&5& 62,551.2&    31,307.8& 62,206.0&  31,101.6\\
						&&0.8	&5&  62,511.0&      50,029.2& 62,246.6&    49,795.2\\
						&0.8&	0.2&	5& 99,836.0&  19,922.6&   99,609.0&  19,920.2\\
						&&	0.5&	5& 99,783.8&   49,714.8&  99,679.0&     49,838.4\\
						&&0.8&5&	  99,861.0&  79,831.4& 99,724.0&   79,775.8\\
						\bottomrule
			\end{tabular}}}
	\end{center}\end{table}
	
	\begin{table}
		\caption{Characteristics of instance set A.2} \label{tab:instA2}
		\begin{center}
			{\scalebox{1.0}{
					\begin{tabular}{rrrrrr}
						\toprule
						$n$&$p_a$&$p_{a'}$  & $\#inst$  & Average number  & Average number \\    
						&& &   & of arcs &  of arcs in $\setap$\\
						\midrule
						100	&0.2	&0.01&5& 994.4&          10.0\\
						&0.5&	0.005&5&	 2,479.8&           10.6\\
						&0.8	&0.001&5&	   3,971.0&          5.0\\
						\midrule	
						300	&0.2	&0.001&	5& 8,966.4&   10.4\\
						&0.5&	0.0003	&5&  22,434.4&     6.8\\
						&0.8	&0.0001&5&	  35,888.4&      6.2\\
						\midrule	
						500	&0.2	&0.0003&5& 24,982.2&  8.0  \\
						&0.5	&0.0003&5& 62,414.6&   17.0\\
						&0.8&	0.0001&	5& 99,821.6&    7.6\\
						\bottomrule
			\end{tabular}}}
	\end{center}\end{table}
	
	\begin{table}
		\caption{Characteristics of instance set B} \label{tab:instAF}
		\begin{center}
			{\scalebox{1.0}{
					\begin{tabular}{rrrrr}
						\toprule
						Group   & $\#inst$  & $n$  & Average number  & Average number \\ 
						&   & & of arcs &  of arcs in $\setap$  \\
						\midrule
						$G_1$  & 5&85 &       2,267.4&       1,480.8 \\
						$G_2$  & 5&150 &       7,168.4&       4,507.4\\
						$G_3$  & 5&160 &        8,269.0&       5,109.6\\
						$G_4$  &5 &210 &  14,326.8&       9,301.2\\
						$G_5$  & 5&240 &  18,843.8&      11,958.4\\
						$G_6$  & 5&270 &        23,611.8&      14,958.4\\
						$G_7$  & 5&300 &     28,721.8&        18,386.0 \\
						$G_8$  & 1&500 &  79,519.0&        52,037.0  \\
						\bottomrule
					\end{tabular}
			}}
		\end{center}
	\end{table}
	
	The second set of instances, B, is derived from the crew scheduling instances at Air France used in~\cite{tellache2024linear}. These instances are grouped into eight categories, each containing five instances, except for the last group, which contains one instance. Within each group, all instances share the same number of pairings (a pairing is a sequence of flights starting and ending at the same crew base) for a given month. Each pairing is characterized by a start and end time. The DAG for each instance is constructed as follows: each pairing is represented by a node, and there is an arc between two nodes, $i$, and $j$, if the end time of $i$ precedes the start time of $j$. Additionally, two nodes are introduced: one representing the start of the month, with arcs directed from this node to all other nodes, and another representing the end of the month, with arcs directed from all other nodes to this end-month node. The arcs in $\setap$ are those in $A$ spanning more than seven consecutive days. Table \ref{tab:instAF} summarizes the characteristics of set B, which contains 36 instances.
	
	We observed that all instances in set B have a DAG $G=(V,A)$ with density around $60\%$ (calculated as $100 \times |A|/(|V|(|V|-1)/2)$) and a sparsity of $\setap$ around $60\%$ (calculated as $100 \times (|\setap|/|A|)$). To evaluate the performance of our algorithms on instances with comparable characteristics but varying DAG densities and different levels of sparsity of $\setap$, we generated the instances of set C. These instances were designed to simulate the airline crew scheduling problem. Initially, time intervals were generated, each with start and end times within a one-month horizon. The duration of these intervals was adjusted according to the desired DAG density: lower densities required longer intervals, resulting in fewer arcs between the corresponding nodes. We began with shorter span durations for the arcs in $\setap$ and incrementally increased this value until the required sparsity was achieved. As with set A, we generated five instances for each combination of $(n, p_{a}, p_{a'})$, considering $n \in \{100, 300, 500\}$ and $p_{a}, p_{a'} \in \{0.2, 0.5, 0.8\}$, resulting in a total of 135 instances. The characteristics of the instances of this set are summarized in Table~\ref{tab:instA1C}. 
	
	In the rest of this section, we refer to different groups of instances using the notation $(n, p_{a}, p_{a'})$. Note that graphs with parameters $p_{a}$ and $ p_{a'}$ have a density around $(100 \times p_{a})\%$ and a sparsity of $\setap$ around $(100 \times p_{a'})\%$.
	
	To summarize, the main features of the generated set of instances are as follows:
	\begin{itemize} 
		\item Set A.1: 135 randomly generated using a uniform distribution, with arc density and sparsity of $\setap$ approximately set to $\{20\%, 50\%, 80\%\}$. 
		\item Set A.2: 45 randomly generated using a uniform distribution, with arc density values around ${20\%, 50\%, 80\%}$ and very low sparsity in $\setap$, where the number of arcs ranges from 3 to 19. 
		\item Set B: 36 real-world instances with transitive DAGs, where arcs in $\setap$ have a specific distribution in the DAG, and the arc density and sparsity of $\setap$ are around 60\%. 
		\item Set C: 135 randomly generated to resemble the instances in set B, with arc density and sparsity of $\setap$ set to values around $\{20\%, 50\%, 80\%\}$.
	\end{itemize}
	
	\subsection{Results} \label{sec:results}
	
	This section presents the results of the BC algorithm applied to the three sets of instances. Section~\ref{sec:expbcroot} discusses the effectiveness of the various cuts applied at the root node across all sets, while Section~\ref{sec:expbc} details the results of the best-performing BC variant for each set.
	
	Let $\overline{\mFB}$ represent formulation $\mFB$ without constraints~\eqref{F3R:cons-inf}. By removing these constraints, $\overline{\mFB}$ reduces the problem to covering a maximum number of nodes using a minimum number of paths. We further define $\overline{L\mFB}$[+x] (respectively $\overline{\mFB}$[+x]) as the LP relaxation of $\overline{\mFB}$ (respectively $\overline{\mFB}$) with the addition of cuts from the first class to class x, following the specified order.
	
	The results presented in this section for variants incorporating A-GRC and RC(\(\pm\), \(-\), \(+\)) were obtained using the greedy algorithm GWMIN~\cite{SMK03} to solve the maximum weighted independent set problem. Additional results are provided in the Appendix \ref{sec:appendix} for comparison using an exact algorithm from Cliquer~\cite{niskanen2003cliquer}. We set a one-hour time limit for all runs. 
	
	\begin{table}
		\caption{Lower bounds and effectiveness of the different cuts on set A.1} \label{tab:lbsA1}
		\begin{center}
			\setlength{\tabcolsep}{0pt}
			\begin{tabular*}{\textwidth}{@{\extracolsep{\fill}}rrrrrrrr}
				\toprule
				&&&\multicolumn{2}{c}{ $\overline{L\mFB}$} & \multicolumn{3}{c}{ $\overline{L\mFB}$[+IPCs,TC-I,II]}  \\
				\cmidrule{4-5}\cmidrule{6-8}
				$n$&$p_a$&$p_{a'}$   & $t(s)$ &$\%best-gap$  & $\#cuts$ & $t(s)$   & $\%best-gap$     \\
				\midrule
				100	&0.2	&0.2&	0.006&	0&	55.4&	0.023&	0\\
				&&0.5&	0.005&	0	&30.4&	0.01&	0\\
				&&0.8&	0.005&	0	&13.0	&0.011&	0\\
				&0.5&	0.2	&0.008&	0	&43.6&	0.017&	0\\
				&&0.5	&0.008&	0	&24.6&	0.04&	0\\
				&&0.8	&0.008&	0	&11.0&	0.043&	0\\
				&0.8&	0.2	&0.011	&0	&28.4	&0.055	&0\\
				&&0.5&	0.011&	0	&18.0&	0.068&	0\\
				&&0.8&	0.011&	0	&9.8&	0.026&	0\\
				\midrule
				300	&0.2&	0.2&	0.035&	0	&110.4&	0.069&	0\\
				&&0.5	&0.033&	0	&70&	0.068&	0\\
				&&0.8	&0.034&	0	&31.8&	0.078&	0\\
				&0.5	&0.2	&0.061&	0	&44	&0.858&	0\\
				&&0.5&	0.064&	0	&26.0&	0.835&	0\\
				&&0.8&	0.057&	0	&13.2	&0.894&	0\\
				&0.8	&0.2&	0.083&	0	&26.2&	1.986&	0\\
				&&0.5&	0.079&	0	&16.8&	1.181&	0\\
				&&0.8&	0.08&	0	&10.2&	1.062&	0\\
				\midrule
				500	&0.2	&0.2	&0.087&	0&	107.8&	0.815&	0\\
				&&0.5&	0.086&	0	&69.2&	0.968&	0\\
				&&0.8&	0.086&	0	&30.6&	0.479&	0\\
				&0.5&	0.2&	0.223	&0&	42.0&	1.044&	0\\
				&&0.5	&0.216&	0	&26.8	&0.647&	0\\
				&&0.8	&0.219&	0	&15.0&	1.609&	0\\
				&0.8&	0.2&	0.313&	0	&27.0&	3.139&	0\\
				&&	0.5&	0.337&	0	&16.6&	3.436&	0\\
				&&0.8	&0.281&	0	&7.4	&3.061&	0\\
				\bottomrule
			\end{tabular*}
	\end{center}\end{table}
	
	\begin{table}
			\scriptsize
		\caption{Lower bounds and effectiveness of the different cuts on set A.2} \label{tab:lbsA2}
		\begin{center}
			\setlength{\tabcolsep}{0pt}
			\begin{tabular*}{\textwidth}{@{\extracolsep{\fill}}rrrrrrrrrrrrrr}
				\toprule
				&&&\multicolumn{2}{c}{ $\overline{L\mFB}$} & \multicolumn{3}{c}{ $\overline{L\mFB}$[+IPCs,TC-I,II]} & \multicolumn{3}{c}{ $\overline{L\mFB}$[+A-GRC]} & \multicolumn{3}{c}{ $\overline{L\mFB}$[+RC,$\pm,-,+$]} \\
				\cmidrule{4-5}\cmidrule{6-8} \cmidrule{9-11}\cmidrule{12-14} 
				$n$     & $p_a$ & $p_{a'}$& $t(s)$ &$\%best-gap$  & $\#cuts$ & $t(s)$   & $\%best-gap$  & $\#cuts$  & $t(s)$   & $\%best-gap$   & $\#cuts$ & $t(s)$   & $\%best-gap$      \\
				\midrule
				100	&0.2	&0.01	&0.005	&16.196&	113.0&	0.092	&775.743	&113.0&	0.094	&775.743	&118.8&	0.099	&\textbf{181.741}\\
				&0.5&	0.005&	0.008	&3.209&	119.0	&0.096	&\textbf{237.469}	&119.0&	0.098&	\textbf{237.469}&	123.6&	0.110	&\textbf{237.469}\\
				&0.8	&0.001&	0.011	&3.401	&83.4	&0.080&	164.321&	83.4&	0.083	&164.321	&95.4	&0.088	&\textbf{132.319}\\
				\midrule	
				300	&0.2	&0.001&	0.034	&12.999&	203.0	&0.326&	496.499&	203.0&	0.351	&496.499&	199.8	&0.329&	\textbf{497.269}\\
				&0.5&	0.0003	&0.069&	10.068	&130.6	&0.614&	\textbf{176.370}&	130.6&	0.657&	\textbf{176.370}&138.0	&0.730	&\textbf{176.370}\\
				&0.8	&0.0001&	0.080&	1.867	&233.8&	0.693&	\textbf{42.719}&	233.8&	0.727	&\textbf{42.719}&	266.4&	0.765&	\textbf{42.719}\\
				\midrule	
				500	&0.2	&0.0003&	0.091	&21.559&	180.8	&0.664	&280.197&	180.8	&0.796&	280.197&	180.2&	0.879	&\textbf{45.011}\\
				&0.5	&0.0003&	0.228&	3.161	&250.2&	3.193	&\textbf{34.874}&	250.2	&3.492&	\textbf{34.874}&	250.2&	3.641	&\textbf{34.874}\\
				&0.8&	0.0001&	0.326&	0.520	&263.6	&2.404	&\textbf{43.060}	&263.6&	2.640	&\textbf{43.060}&	386.8	&3.528	&43.235\\
				
				\bottomrule
			\end{tabular*}
	\end{center}\end{table}
	
	\begin{table}
		\caption{Lower bounds and effectiveness of the different cuts on set B} \label{tab:lbsB}
		\begin{center}
			
			\setlength{\tabcolsep}{0pt}
			\begin{tabular*}{\textwidth}{@{\extracolsep{\fill}}rrrrrr}
				\toprule
				&\multicolumn{2}{c}{ $\overline{L\mFB}$} & \multicolumn{3}{c}{ $\overline{L\mFB}$[+IPCs,TC-I,II]}  \\
				\cmidrule{2-3}\cmidrule{4-6}
				Groups    & $t(s)$ &$\%best-gap$  & $\#cuts$ & $t(s)$   & $\%best-gap$     \\
				\midrule
				$G_1$&	0.013&	1.196&	82.6	&0.223&	117.820\\
				$G_2$&	0.012&	0.005&	167.8	&0.593&	16.623\\
				$G_3$&	0.013&	1.143&	247.4	&0.907&	16.988\\
				$G_4$&	0.026&	5.244&	200.0  &	1.608&	16.665\\
				$G_5$&	0.027&	0.008&	299.0&	1.675	&23.713\\
				$G_6$	&0.034&	3.784&	269.4&	1.412	&20.782\\
				$G_7$	&0.043&	2.138&	442.8&	3.107&	22.955\\
				$G_8$	&0.121&	20.021&	105.0&	4.754&	-\\
				
				\bottomrule
			\end{tabular*}
	\end{center}\end{table}

	\begin{table}
			\scriptsize
		\caption{Lower bounds and effectiveness of the different cuts on set C} \label{tab:lbsC}
		\begin{center}
			
			\setlength{\tabcolsep}{0pt}
			\begin{tabular*}{\textwidth}{@{\extracolsep{\fill}}rrrrrrrrrrrrrr}
				\toprule
				&&&\multicolumn{2}{c}{ $\overline{L\mFB}$} & \multicolumn{3}{c}{ $\overline{L\mFB}$[+IPCs,TC-I,II]} & \multicolumn{3}{c}{ $\overline{L\mFB}$[+A-GRC]} & \multicolumn{3}{c}{ $\overline{L\mFB}$[+RC,$\pm,-,+$]} \\
				\cmidrule{4-5}\cmidrule{6-8} \cmidrule{9-11}\cmidrule{12-14} 
				$n$     & $p_a$ & $p_{a'}$& $t(s)$ &$\%best-gap$  & $\#cuts$ & $t(s)$   & $\%best-gap$  & $\#cuts$  & $t(s)$   & $\%best-gap$   & $\#cuts$ & $t(s)$   & $\%best-gap$      \\
				\midrule
				100	&0.2	&0.2&	0.005&	68.326	&85.0&	0.101	&\textbf{0.000}	&85.0&	0.112&	\textbf{0.000}	&79.8	&0.024&	\textbf{0.000}\\
				&&	0.5&	0.005&	48.298&	96.0&	0.165&	29.252&	96.0&	0.188&	29.252	&102.0&	0.058&	\textbf{2.126}\\
				&&0.8&	0.005	&38.483&	99.0&	0.187&	2.785&	99.0&	0.184&	2.785&	102.8	&0.062&	\textbf{1.105}\\
				&0.5&	0.2&	0.009&	44.213&	68.8&	0.130&	\textbf{0.000}&	68.8&	0.215&	\textbf{0.000}	&73.2&	0.062&	2.696\\
				&&0.5	&0.009&	13.042&	139.0&	0.140&	\textbf{90.659}&	139.0&	0.140&	\textbf{90.659}&	138.2	&0.140	&\textbf{90.659}\\
				&&0.8	&0.009&	0.004&	100.4&	0.138&	\textbf{0.000}&100.4&	0.094&	\textbf{0.000}	&108.0	&0.292&	-\\
				&0.8	&0.2	&0.012&	22.120&	37.6&	0.081&	\textbf{0.000}	&37.6&	0.026&	\textbf{0.000}	&37.6&	0.044&	\textbf{0.000}\\
				&&0.5	&0.011&	0.022&	111.4&	0.150&	\textbf{22.021}&	111.4&	0.142&	\textbf{22.021}&	111.4&	0.485	&\textbf{22.021}\\
				&&0.8	&0.011&	0.000	&92.2&	0.116&	\textbf{22.488}&	92.2	&0.119&	\textbf{22.488}&	92.2	&0.385&	\textbf{22.488}\\
				\midrule
				300	&0.2	&0.2	&0.032&	65.177&	221.8	&0.708&	6.295&	221.8&	0.455	&6.295&	230.4&	0.275	&\textbf{4.243}\\
				&&0.5	&0.035&	45.434&	301.0&	0.858&	\textbf{10.162}&	301.0&	0.317&	\textbf{10.162}&	301.0&	0.505&	\textbf{10.162}\\
				&&0.8	&0.035&	33.833&	261.6&	1.524	&5.468&	261.6	&0.398&	5.468	&263.0&	0.401&	\textbf{4.764}\\
				&0.5	&0.2	&0.051&	39.741&	121.8	&1.470&	3.283&	121.8&	0.438&	3.283	&114.0&	1.600&	\textbf{0.572}\\
				&&0.5&	0.050&	9.286&	278.4&	2.408&	\textbf{13.528}&	278.4&	0.831&	\textbf{13.528}&	278.4&	1.614&	\textbf{13.528}\\
				&&0.8	&0.059&	2.131&	252.8&	1.779	&\textbf{21.332}&	252.8&	0.663&	\textbf{21.332}&	252.8&	1.927&	\textbf{21.332}\\
				&0.8	&0.2	&0.077&	19.034&	60.0&	0.370&	\textbf{0.000}&	60.0&	0.346&	\textbf{0.000}	&56.4&	0.438&	\textbf{0.000}\\
				&&0.5	&0.076&	13.478&	213.8&	3.922&	\textbf{17.006}&	213.8	&2.369&	\textbf{17.006}&	213.8	&1.395&	\textbf{17.006}\\
				&&0.8&	0.076&	0.133	&193.6	&3.451	&\textbf{31.006}&	193.6&	2.550	&\textbf{31.006}&	193.6&	1.873&	\textbf{31.006}\\
				\midrule
				500	&0.2&	0.2&	0.057&	66.027&	309.2	&2.245&	2.180&	309.2&	2.456&	2.180&	302.8&	2.383	&\textbf{1.536}\\
				&&	0.5&	0.057&	44.860&	415.4&	2.699&	\textbf{14.395}&	415.4&	0.927&	\textbf{14.395}&	415.4&	1.845&	\textbf{14.395}\\
				&&0.8	&0.059&	30.621&	369.8&	2.113&	\textbf{7.027}&	369.8&	1.856&	\textbf{7.027}&	369.8	&2.273	&\textbf{7.027}\\
				&0.5	&0.2	&0.166&	37.645&	164.2	&4.015&	\textbf{0.653}	&164.2&	1.683&	\textbf{0.653}	&171.6	&2.472	&1.869\\
				&&0.5	&0.156&	10.379&	422.8&	6.350&	\textbf{7.199}&	422.8&	6.350&	\textbf{7.199}&	422.8&	2.693	&\textbf{7.199}\\
				&&0.8	&0.155	&11.839&	331.6	&5.145&	\textbf{9.650}&	331.6&	2.109	&\textbf{9.650}&331.6	&1.904	&\textbf{9.650}\\
				&0.8&	0.2	&0.264	&19.982&	66.8&	1.394&	\textbf{0.000}	&66.8	&0.592&	\textbf{0.000}	&66.8	&0.522&	\textbf{0.000}\\
				&&0.5	&0.257	&18.286&	298.0&	10.563&	\textbf{4.210}&	298.0&	4.409&	\textbf{4.210}&	298.0&	3.638&	\textbf{4.210}\\
				&&0.8	&0.264&	27.361	&292.4	&16.735&-	&292.4	&7.367&	-	&	292.4&5.939&	-\\
				
				\bottomrule
			\end{tabular*}%
	\end{center}\end{table}
	
	\subsubsection{Performance of the BC algorithm at the root node} \label{sec:expbcroot}
	
	This section presents the results of applying different cut classes to the LP relaxation of formulation \mFB. 
	
	Tables~\ref{tab:lbsA1}-\ref{tab:lbsC} summarize the results for sets A.1, A.2, B, and C, respectively. The metrics presented in these tables are as follows: $t(s)$ represents the average CPU time in seconds; $\#cuts$ indicates the total number of cuts added; and $\%best-gap$ denotes the average gap concerning the best-known solution for each instance. This gap is calculated as$(|(z^\ast - LB) / LB|) \times 100$, where $LB$ is the objective value returned by $\overline{L\mFB}$\ or $\overline{L\mFB}$[+x], and  $z^\ast$ is the best-known objective value obtained from all the runs of $\overline{L\mFB}$[+x] and $\overline{\mFB}$[+x]. In the tables, bold values indicate the best gap values.
	
	Table~\ref{tab:lbsA1} presents the results for set A.1. The $\%best-gap$ is null for $\overline{L\mFB}$\ across all instances, indicating that the optimal objective values match the best-known objective values. However, the corresponding solutions are not necessarily feasible for \mFB\ with respect to the constraints related to $\setap$, as $\overline{L\mFB}$[+IPCs,TC-I,II] added cuts to enforce feasibility without increasing the objective function value (the same number of selected paths and covered nodes, see Tables~\ref{tab:LF2A1} and~\ref{tab:LBv1A1} for the corresponding details). The number of these cuts decreases as the values of $p_a$ and $p_{a'}$ increase. This can be explained by the fact that when the number of arcs in $A$ and  $\setap$ increases, the likelihood that the paths selected by \mFB\ are already feasible also increases, thereby reducing the number of cuts needed to make those paths feasible. By incorporating the classes IPC, TC-I, and TC-II, all instances were solved to optimality in an average of less than 4 seconds. Since this variant achieved optimality for all instances, we did not include results for other variants with additional classes of cuts, as their separation procedures increased CPU time without enhancing the already optimal solutions. Further tests on larger instances with 1000 and 2000 nodes showed that the algorithm maintained its efficiency, solving all instances to optimality in under 30 seconds on average.

	For sets A.2 and C, $\overline{L\mFB}$[+IPCs,TC-I,II] and $\overline{L\mFB}$[+A-GRC] variants yield identical solutions across all instances (with the same $\%best-gap$), and the number of added cuts remains the same, indicating that no A-GRC cuts have been added (see Tables~\ref{tab:lbsA2} and~\ref{tab:lbsC}). In contrast, $\overline{L\mFB}$[+RC,$\pm,-,+$] demonstrates improvements across most instance classes, albeit with a slight increase in CPU time (see Tables~\ref{tab:compareLBA2} and~\ref{tab:compareLBC} in Appendix \ref{sec:appendix} for more details on the comparisons). For set A.2, all three variants solved 3 out of 45 instances to optimality at the root node (see Tables~\ref{tab:LBv1A2}-\ref{tab:LBv3CA2}). For set C, $\overline{L\mFB}$[+IPCs,TC-I,II] and $\overline{L\mFB}$[+A-GRC] solved 28 instances to optimality, while $\overline{L\mFB}$[+RC,$\pm,-,+$] solved to optimality 35 out of 135 instances (see Tables~\ref{tab:LBv1C}-~\ref{tab:LBv3cC}). 
	
	For set B, only $\overline{L\mFB}$[+IPCs,TC-I,II] results are presented in Table~\ref{tab:lbsB}, as the other variants produced the same values for all the criteria but generally required more CPU time due to the more time-consuming separation procedures (see Table~\ref{tab:compareLBB}). None of the instances were solved to optimality at the root node, and 22 out of 36 instances had a null objective value, meaning the initial solution with no path selected and no node covered could not be improved at the root node in these cases (see Table~\ref{tab:LBv1B}).
	
	Regarding the quality of the solutions produced, we observe that the value of $\%best-gap$ for $\overline{L\mFB}$\ generally decreases as the values of $p_a$ and $p_{a'}$ increase. This is because the number of arcs in $A$ and $\setap$ increases, enhancing the likelihood that the paths selected by $\overline{L\mFB}$\ are feasible for \mFB. The $\%best-gap$ values produced by $\overline{L\mFB}$\ are better than those generated by $\overline{L\mFB}$[+x] for sets A.2 and B but are generally higher than those produced by $\overline{L\mFB}$\ for set C.
	
	The best-performing variant at the root node can be summarized as follows:
	\begin{itemize} 
		\item Set A.1: $\overline{L\mFB}$[+IPCs,TC-I,II], with all 135 instances solved to optimality.
		\item Set A.2: $\overline{L\mFB}$[+RC,$\pm,-,+$], with 3 instances out of 45 solved to optimality and an average gap of $154.556\%$ from the best-known solution.
		\item Set B: $\overline{L\mFB}$[+IPCs,TC-I,II], with none of the 36 instances solved to optimality and an average gap of $29.443\%$ from the best-known solution.
		\item Set C: $\overline{L\mFB}$[+RC,$\pm,-,+$], with 36 instances out of 135 solved to optimality and an average gap of $11.786\%$ from the best-known solution.
	\end{itemize}
	
	Further details on the results are provided in the Appendix \ref{sec:appendix}. Tables~\ref{tab:LF2A1}-\ref{tab:LF2C} present the results of $\overline{L\mFB}$\ for sets A.1, A.2, B, and C, respectively. The results of $\overline{L\mFB}$[+IPCs,TC-I,II] for set A.1 are shown in Table \ref{tab:LBv1A1}, while those for set A.2 across the three variants are presented in Tables~\ref{tab:LBv1A2}-\ref{tab:LBv3CA2}. For set B, the results of $\overline{L\mFB}$[+IPCs,TC-I,II] are detailed in Table \ref{tab:LBv1B}. Finally, the results for set C across the three variants are provided in Tables~\ref{tab:LBv1C}-~\ref{tab:LBv3cC}. More details on the comparison between the lower bounds are presented in Tables~\ref{tab:compareLBA2},~\ref{tab:compareLBB}, and~\ref{tab:compareLBC} for sets A.2, B, and C, respectively.
	
	\subsubsection{Performance of the BC algorithm with different cut configurations} \label{sec:expbc}
	
	This section evaluates BC algorithm variants with alternative cut configurations to assess their combined effectiveness.
	
	Since all instances in set A.1 were solved to optimality at the root node, we applied the BC variants $\overline{\mFB}$[+IPCs,TC-I,II],  $\overline{\mFB}$[+A-GRC], and  $\overline{L\mFB}$[+RC,$\pm,-,+$] to the remaining sets: A.2, B, and C. Detailed performance results are presented in the Appendix: Tables~\ref{tab:BCv1A2}-\ref{tab:BCv3cA2} for set A.2, Tables~\ref{tab:BCv1B}-\ref{tab:BCv3cB} for set B, and Tables~\ref{tab:BCv1C}-\ref{tab:BCv3cC} for set C. Comparative summaries for each set are also provided: Table~\ref{tab:compareBCA2} for set A.2, Table~\ref{tab:compareBCB} for set B, and Table~\ref{tab:compareBCC} for set C. 
	
	The results of the best-performing BC variant for each set are highlighted in Tables~\ref{tab:BCA2}-\ref{tab:BCC}, corresponding to sets A.2, B, and C, respectively. These tables include the metrics discussed in previous sections, along with: $\#opt$, the number of instances solved to optimality; $\#opt_0$, the number of instances solved to optimality at the root node; $grb-nodes$, the average number of nodes generated by Gurobi; and $t_{\text{sep}}(s)$, the average CPU time (in seconds) spent in the separation procedures.
	
	As seen in the previous section, the BC algorithm exhibited strong performance on set A.1, solving all instances to optimality at the root node in under 4 seconds (see Table~\ref{tab:LBv1A1} for details) and solving similar instances with 1000 and 2000 nodes to optimality in less than 30 seconds on average. To further evaluate its performance, we compared it against formulation $\mFA$, solved using Gurobi. $\mFA$ solved all instances of A.1 with $n=100$ to optimality in less than 237 seconds on average. However, as the number of nodes increased, $\mFA$ was able to solve only two out of five instances to optimality for the group $(300, 0.2, 0.2)$. For the remaining groups, the program terminated due to insufficient memory (see Table~\ref{tab:F1A1} for more details on $\mFA$). Therefore, the BC algorithm clearly outperforms $\mFA$ on set A.1. We attribute the strong performance of the BC algorithm in this set of instances to the fact that there are enough arcs in $\setap$, uniformly distributed across the graph, making it possible to transform an MPC into a feasible solution without reducing the number of covered nodes or increasing the number of paths. Note that all nodes are covered in the optimal solution for the instances in A.1.
	
	The results for set A.2 indicate that the variant incorporating all classes of cuts in sequence performed the best in most cases (see Table~\ref{tab:compareBCA2}). Accordingly, Table~\ref{tab:BCA2} presents the results for this variant. We observed from the implementations that the instances in this set were the most challenging for the BC variants compared to the other sets. The selected variant solved 12 out of 45 instances to optimality, while $\mFA$ solved all instances to optimality in under 328 seconds on average, except for four instances in the group $(500, 0.5, 0.0003)$ (see Table~\ref{tab:F1A2}). For these four instances, the average Gurobi optimality gap was $3.454\%$, with a $\%best-gap$ of zero, indicating that $\mFA$ yielded superior solutions compared to the BC variants. We think the difficulty of these instances arises from the fact that there are few arcs in $\setap$, which complicates the task of covering all nodes in the optimal solution. This, in turn, makes it challenging for the BC algorithm to transform the initial root-node solution---one that covers all nodes---into a feasible solution with paths containing at least one arc from $\setap$ (see Table~\ref{tab:LF2A2} for the results of $\overline{L\mFB}$). It is also worth noting that the difficult instances were chosen from a pool of 150 generated instances for each group, with the five most challenging selected, and many of these 150 instances were solved quickly to optimality by the BC variants.
	
	\begin{table}
		\scriptsize
		\setlength{\tabcolsep}{0pt}
		\caption{\label{tab:BCA2} Results of $\overline{\mFB}$[+RC,$\pm,-,+$] on A.2.}
		\begin{center}
			\begin{tabular*}{\textwidth}{@{\extracolsep{\fill}}rrrrrrrrrrrr}
				\toprule
				$n$&$p_{a}$&$p_{a'}$&$\#inst$&$\#opt$&$\#opt_0$&$\#null-obj$&$\#cuts$&$grb-nodes$&$t_{sep}(s)$&$t(s)$&$\%best-gap$\\
				
				\midrule
				100	&	0.2	&	0.01	&5&	2	&	0	&	0	&	646,361.8	&	252,217.0	&	52.533	&	2,171.037	&	0.431	\\
				&	0.5	&	0.005&5	&	0	&	0	&	1	&	3,160,615.4	&	193,089.2	&	7.326	&	3,600.703	&	4.608	\\
				&	0.8	&	0.001	&5&	3	&	1	&	0	&	1,434,728.8	&	41,541.2	&	73.617	&	1,459.996	&	98.592	\\
				\midrule
				300	&	0.2	&	0.001	&5&	1	&	0	&	0	&	1,123,671.2	&	140,219.4	&	132.907	&	2,881.114	&	55.867	\\
				&	0.5	&	0.0003	&5&	1	&	0	&	0	&	3,639,333.6	&	73,126.2	&	454.872	&	2,888.236	&	50.280	\\
				&	0.8	&	0.0001	&5&	3	&	1	&	0	&	2,641,659.6	&	19,376.4	&	960.649	&	2,086.653	&	11.674	\\
				\midrule
				500	&	0.2	&	0.0003	&5&	0	&	0	&	0	&	1,097,705.6	&	108,667.4	&	284.691	&	3,603.035	&	35.316	\\
				&	0.5	&	0.0003	&5&	0	&	0	&	0	&	4,285,981.2	&	40,196.6	&	982.717	&	3,619.633	&	1,278.362	\\
				&	0.8	&	0.0001	&5&	2	&	1	&	1	&	1,669,467.4	&	6,628.4	&	1,759.262	&	2,562.877	&	84.663	\\
				
				\bottomrule
			\end{tabular*}
		\end{center}
	\end{table}
	
	\begin{table}
		\scriptsize
		\setlength{\tabcolsep}{0pt}
		\caption{\label{tab:BCB} Results of $\overline{\mFB}$[+IPCs,TC-I,II] on B.}
		\begin{center}
			\begin{tabular*}{\textwidth}{@{\extracolsep{\fill}}rrrrrrrrr}
				\toprule
				Group&$\#inst$&$\#opt$&$\#opt_0$&$\#cuts$&$grb-nodes$&$t_{sep}(s)$&$t(s)$&$\%best-gap$\\
				
				\midrule
				$G_1$	&	5&5	&	0	&	2,116.6	&	1,062.2	&	0	&	0.321	&	0	\\
				$G_2$	&	5&5	&	0	&	25,052.6	&	11,273.6	&	0	&	14.919	&0	\\
				$G_3$	&	5&4	&	0	&	205,428.2	&	239,446.6	&	0	&	744.497	&	0	\\
				$G_4$	&	5&4	&	0	&	627,299.0	&	299,506.8	&	0.004	&	1,329.738	&	3.222	\\
				$G_5$	&	5&4	&	0	&	238,346.4	&	111,251.8	&	0.001	&	771.869	&	0	\\
				$G_6$	&	5&2	&	0	&	801,363.8	&	383,725.6	&	0.003	&	2,303.893	&	0.16	\\
				$G_7$	&	5&4	&	0	&	116,830.8	&	94,327.2	&	0.131	&	928.903	&	0	\\
				$G_8$	&	1&0	&	0	&	153,998.0	&	68,688.0	&	0.007	&	3,601.425	&	0	\\
				
				\bottomrule
			\end{tabular*}
		\end{center}
	\end{table}
	
	\begin{table}
		\scriptsize
		\setlength{\tabcolsep}{0pt}
		\caption{\label{tab:BCC} Results of $\overline{\mFB}$[+IPCs,TC-I,II] on C.}
		\begin{center}
			\begin{tabular*}{\textwidth}{@{\extracolsep{\fill}}rrrrrrrrrrr}
				\toprule
				$n$&$p_{a}$&$p_{a'}$&$\#inst$&$\#opt$&$\#opt_0$&$\#cuts$&$grb-nodes$&$t_{sep}(s)$&$t(s)$&$\%best-gap$\\
				
				\midrule
				100	&	0.2	&	0.2	&	5	&	5	&	5	&	85.0	&	1.0	&	0	&	0.033	&	0	\\
				&		&	0.5	&	5	&	5	&	1	&	2,322.2	&	6,430.8	&	0	&	1.159	&	0	\\
				&		&	0.8	&	5	&	5	&	2	&	129.2	&	229.0	&	0	&	0.096	&	0	\\
				&	0.5	&	0.2	&	5	&	5	&	4	&	554.8	&	199.8	&	0	&	0.134	&	0	\\
				&		&	0.5	&	5	&	4	&	0	&	1,751,951.6	&	2,259,141.2	&	0.412	&	724.581	&	0	\\
				&		&	0.8	&	5	&	5	&	1	&	3,096.6	&	1315.6	&	0	&	0.797	&	0	\\
				&	0.8	&	0.2	&	5	&	5	&	5	&	37.6	&	1.0	&	0	&	0.022	&	0	\\
				&		&	0.5	&	5	&	5	&	0	&	38,698.2	&	4,037.6	&	0.009	&	4.851	&	0	\\
				&		&	0.8	&	5	&	5	&	1	&	453.0	&	67.0	&	0	&	0.174	&	0	\\
				\midrule
				300	&	0.2	&	0.2	&	5	&	5	&	0	&	2,866.6	&	952.8	&	0	&	1.585	&	0	\\
				&		&	0.5	&	5	&	5	&	0	&	48,624.0	&	22,662.0	&	0	&	34.312	&	0	\\
				&		&	0.8	&	5	&	5	&	0	&	2012.0	&	1,216.0	&	0	&	2.037	&	0	\\
				&	0.5	&	0.2	&	5	&	4	&	1	&	225,028.0	&	128,198.4	&	0.001	&	724.553	&	0	\\
				&		&	0.5	&	5	&	3	&	0	&	1,798,611.2	&	327,273.4	&	0.132	&	1,535.603	&	0.457	\\
				&		&	0.8	&	5	&	3	&	0	&	539,172.6	&	216,082.2	&	0.006	&	1,761.739	&	0.136	\\
				&	0.8	&	0.2	&	5	&	5	&	3	&	2,505.0	&	393.6	&	0	&	2.663	&	0	\\
				&		&	0.5	&	5	&	0	&	0	&	1,635,810.4	&	255,491.2	&	0.037	&	3,600.442	&	0	\\
				&		&	0.8	&	5	&	4	&	0	&	573,132.6	&	77,184.4	&	0.004	&	774.897	&	0	\\
				\midrule
				500	&	0.2	&	0.2	&	5	&	2	&	0	&	2,102,109.4	&	682,378.4	&	0.017	&	2,274.351	&	0.122	\\
				&		&	0.5	&	5	&	4	&	0	&	486,017.2	&	504,971.0	&	0.076	&	1,383.979	&	0	\\
				&		&	0.8	&	5	&	3	&	0	&	631,104.4	&	564,780.0&	0.015	&	1,448.953	&	0	\\
				&	0.5	&	0.2	&	5	&	2	&	0	&	1,150,957.8	&	390,233.8	&	0.046	&	2,168.779	&	0	\\
				&		&	0.5	&	5	&	0	&	0	&	1,186,512.0	&	250,625.0	&	0.602	&	3,600.474	&	0.655	\\
				&		&	0.8	&	5	&	0	&	0	&	392,207.6	&	150,823.4	&	0.017	&	3,600.654	&	0.898	\\
				&	0.8	&	0.2	&	5	&	5	&	5	&	66.8	    &	1.0	&	0.002	&	0.525	&	0	\\
				&		&	0.5	&	5	&	0	&	0	&	550,922.4	&	52,420.8	&	0.099	&	3,600.967	&	0	\\
				&		&	0.8	&	5	&	1	&	0	&	431,369.2	&	34,284.0	&	0.146	&	2,776.744	&	0	\\
				
				\bottomrule
			\end{tabular*}
		\end{center}
	\end{table}
	
	For set B, $\overline{\mFB}$[+IPCs,TC-I,II] performed the best on average (see Table~\ref{tab:compareBCB}). Note that the BC algorithm was terminated for groups $G_6$ and $G_7$ before reaching the one-hour time limit for some instances in the third variant. Therefore, we reduced the time limit to 1800 seconds for these cases, reflected in the gaps for groups $G_6$ and $G_7$. The results for $\overline{\mFB}$[+IPCs,TC-I,II] are summarized in Table~\ref{tab:BCB}. This variant solved 28 out of 36 instances to optimality, with an average $\%best-gap$ of $0.423\%$. In comparison, $\mFA$ solved $3$ out of $5$ instances in $G_1$ to optimality, while the remaining groups experienced program termination due to insufficient memory (see Table~\ref{tab:F1B}). Note that $G_1$ contains $n=85$, with $p_{a} \approx 0.6$ and $p_{a'} \approx 0.6$. In this set, not all nodes can be covered in an optimal solution, even though $\setap$ has a sparsity of approximately $60\%$. This is because of the particular distribution of the arcs in $\setap$: If an arc $(i, j)$ is in $\setap$, all arcs $(k, l)$ where $k$ precedes or coincides with $i$, and $l$ succeeds or coincides with $j$, are also included in the DAG, which has a transitive structure.
	
	For set C, $\overline{\mFB}$[+IPCs,TC-I,II] also performed the best on average (see Table~\ref{tab:compareBCC}). This variant solved 95 out of 135 instances to optimality, with an average $\%best-gap$ of $0.084\%$. Recall that the instances in set C share the same characteristics as those in set B, and similar to set B, not all nodes can be covered in an optimal solution despite the sparsity of $\setap$ being greater than $20\%$. We can observe from Table~\ref{tab:BCC} that as the value of $n$ increases, the performance of the BC algorithm decreases. For parameters $p_a$ and $p_{a'}$, instances with values of $0.5$ for both are generally the most challenging. When compared to $\mFA$, this model solved 27 instances to optimality, including all instances with $n=100$ and $p_{a'}=0.2$ (see Table~\ref{tab:F1C}). Furthermore,  given the null $\%best-gap$ for instances with $n=100$ and $p_{a}=0.2$ and the fact that these instances were all solved to optimality by the BC variant, the solutions returned by  $\mFA$ for these instances are also optimal. For the remaining instances not listed in Table~\ref{tab:F1C}, the program terminated prematurely due to insufficient memory.
	
	When comparing the results of $\mFA$ with the BC variants across the three instance sets, it appears that the performance of $\mFA$ is primarily influenced by the size of the problem, which increases with higher values of $n$, $p_a$, and $p_{a'}$. In contrast, the BC variants, in addition to being affected by these parameters, are more sensitive to the structure of the graph and the distribution of the arcs of $\setap$.
	
	The results can be summarized as follows:
	\begin{itemize} 
		\item Set A.1: 100\% of the 135 instances were solved to optimality at the root node in less than 4 seconds on average using the $\overline{\mFB}$[+IPCs,TC-I,II] BC variant;
		\item Set A.2: About 27\% of the 45 instances were solved to optimality using $\overline{\mFB}$[+RC,$\pm,-,+$] in an average CPU time of 2763.698 seconds.
		\item Set B: About 78\% of the 36 instances were solved to optimality using $\overline{\mFB}$[+IPCs,TC-I,II] in an average CPU time of 1211.945 seconds.
		\item Set C: About 70\% of the 135 instances were solved to optimality using $\overline{\mFB}$[+IPCs,TC-I,II] in an average CPU time of 1112.041 seconds.
	\end{itemize}
	
	Across all BC variants and \mFA, 304 out of 351 instances (about 87\%) were solved to optimality. The generated instances, along with the best solutions obtained for each, will be made available on the first author's web page.

	\section{Conclusions}\label{sec:con}
	
	We investigated a variant of the \texttt{MPC} problem on acyclic digraphs (DAGs), where each path was required to include at least one arc from a specified subset of arcs. We established that the feasibility problem was strongly $\mathcal{NP}$-hard for arbitrary DAGs but solvable in polynomial time when the DAG was the transitive closure of a path.
	
	Given the potential infeasibility of the problem, we focused on maximizing the number of nodes covered using the minimum number of node-disjoint paths, ensuring each path contained at least one arc from the specified subset. Two integer programming formulations were introduced for this problem, and we developed several valid inequalities to strengthen the linear programming relaxations. These were applied as cutting planes within a branch-and-cut framework. 
	
	Our approach was implemented and tested on a diverse set of instances, including real-world cases from airline crew scheduling, and demonstrated strong effectiveness in solving the problem. The algorithms showed varying performance depending on the DAG type (transitive or non-transitive), the distribution of the arcs of $\setap$ in the DAG, and the sparsity of $\setap$. Overall, the variant using the classes IPCs, TC-I, and TC-II performed well in the majority of the classes, and these results were further enhanced mainly for A2 and C with the addition of +RC,$\pm,-,+$. Overall, 270 instances out of 351 instances (77\%) were solved to optimality, and if we consider all the BC variants and \mFA, 304 out of 351 instances (87\%) were solved to optimality.
	
	For future work, the complexity of the restricted version of the crew scheduling problem, \texttt{CS-MPC-ARC}, which considers the rest block constraints, remains an open question. Investigating this complexity could further enhance our understanding of the minimal components responsible for the hardness of the crew scheduling problem. It would also be interesting to test the performance of different branching and node selection strategies in the branch-and-cut algorithm. Additionally, exploring branch-and-price algorithms, where the pricing problem is restricted to finding feasible paths and comparing them with the branch-and-cut algorithm, could provide valuable insights.
	
	\section*{Acknowledgments}
	The authors thank Mohand Ait Alamara and Séverine Bonnechère at Air France for providing concrete instances, initially used in~\cite{tellache2024linear}, some of which have also been utilized in this paper.
	
	\bibliographystyle{unsrtnat}
	\bibliography{biblio}
	
	
	\appendix
	
	\bmsection{Detailed computational results}\label{sec:appendix}
	
	This section presents detailed results from the runs of various BC algorithm variants, along with the results of formulation $(\mFA)$. 
	
	Tables~\ref{tab:F1A1}-~\ref{tab:F1C} present the results of formulation  $(\mFA)$, while Tables~\ref{tab:LF2A1}-~\ref{tab:LF2C} show the results of $\overline{L\mFB}$\ across the different sets. The results for various BC variants at the root node on each set of instances are summarized in Tables~\ref{tab:LBv1A1}-\ref{tab:LBv3cC}. Note that for set A.1, we present results for $\overline{L\mFB}$[+IPCs,TC-I,II] only, as it yielded optimal solutions for all instances; testing other variants generally increases CPU time without improving the already optimal solution. For variants using A-GRC and RC (\(\pm\), \(-\), \(+\)), we employed two algorithms to find a maximum weighted independent set, the greedy algorithm GWMIN proposed in~\cite{SMK03} and the exact algorithm provided in Cliquer~\cite{niskanen2003cliquer}. Results for both are shown, with results for variants using Cliquer indicated in parentheses marked as ``(C)''; if not indicated, the GWMIN algorithm was used. For set B, all variants produced identical results, so we provide a single table for the results of $\overline{L\mFB}$[+IPCs,TC-I,II].
	
	Tables~\ref{tab:compareLBA2},~\ref{tab:compareLBB}, and~\ref{tab:compareLBC} present a comparison of the BC variants at the root node, along with the results of  $\overline{L\mFB}$. The metric \( \#bs \) indicates the number of times a given variant achieves the best objective value, while \( \%obj \) quantifies the gap between the average objective value of a variant and the best objective value across all variants, calculated as: $ 
	\%obj = \left| \frac{\sum_{i=1}^5 obj_i - \sum_{i=1}^5 best\_obj_i}{\sum_{i=1}^5 obj_i} \right|,$
	where \( obj_i \) is the objective value for the \( i^{th} \) instance under a given variant, and \( best\_obj_i \) is the best objective value for the \( i^{th} \) instance across all variants. We introduced \( \%obj \) alongside \( \%best-gap \) because the latter does not account for undefined gaps, which occur when the objective value is null. In certain cases, \( \%best-gap \) might suggest that a variant solving only one instance with a non-null objective value performs better than another variant that solves the same instance with the same objective value and an additional instance with a non-null objective value. This happens because \( \%best-gap \) is averaged over fewer instances in the first case.
	
	The results for the BC variants across the three sets of instances are summarized in Tables~\ref{tab:BCv1A2}-\ref{tab:BCv3cA2},\ref{tab:BCv1B}-\ref{tab:BCv3cB}, and\ref{tab:BCv1C}-\ref{tab:BCv3cC}. Comparisons between the BC variants for each set of instances are provided in Tables~\ref{tab:compareBCA2},\ref{tab:compareBCB}, and\ref{tab:compareBCC}.
	
	\begin{table}
		\scriptsize
		\setlength{\tabcolsep}{0pt}
		\caption{\label{tab:F1A1} Results of formulation $(\mFA)$ on  set A.1.}
		\begin{center}
\begin{tabular*}{\textwidth}{@{\extracolsep{\fill}}rrrrrrrrrr}
				\toprule		
				$n$&$p_{a}$&$p_{a'}$&$\#opt$& $t (s)$& $\%grb-gap$ & $grb-nodes$  & $paths$ & $nodes$ & $\%best-gap$\\
				\midrule
				100& 0.2 & 0.2 &    5&      127.628&            0&            0&           9.2&          100&0 \\
				& & 0.5 &   5&      11.116&            0&            0&          9.4&          100&0 \\
				& & 0.8 &   5&       12.708&            0&            0&                    9.4&          100  &0\\
				& 0.5 & 0.2 &     5&      76.746&            0&            0&          4.2&          100 &0 \\
				& & 0.5 &     5&      61.219&            0&            0&                      4.2&          100&0 \\
				& & 0.8 &    5&      106.924&            0&            0&                      4.4&          100 &0\\
				& 0.8& 0.2 &  5&      247.874&            0&            0&                     2.4&          100 &0 \\
				& & 0.5 &  5&      193.384&            0&            0&                     2.8&          100 &0  \\
				& & 0.8 &   5&      236.599&            0&            0&                     2.6&          100 &0 \\
				\midrule
				300& 0.2 & 0.2 &    2&      2,353.310&      11.479&            0&            9.4&        272.8 & 11.479 \\
				\bottomrule
			\end{tabular*}	
		\end{center}
	\end{table}
	
	\begin{table}
		\scriptsize
		\setlength{\tabcolsep}{0pt}
		\caption{\label{tab:F1A2} Results of formulation $(\mFA)$ on  set A.2.}
		\begin{center}
			\begin{tabular*}{\textwidth}{@{\extracolsep{\fill}}rrrrrrrrrr}
				\toprule		
				$n$&$p_{a}$&$p_{a'}$&$\#opt$& $t (s)$& $\%grb-gap$ & $grb-nodes$ & $paths$ & $nodes$ & $\%best-gap$ \\
				\midrule
				100& 0.2 & 0.0100 &    5&      32.182&            0&        232.2&             7.6&         83.8 & 0 \\
				& 0.5 & 0.0050 & 5&      60.629&            0&          7.2&               5&         96.8&0  \\
				& 0.8 &  0.0010  & 5&      1.276&            0&            0&            2.8&         96.6  &0 \\
				\midrule
				300& 0.2 & 0.0010 &   5&      288.047&            0&         62.6&             8.8&          261.0&0 \\
				& 0.5 & 0.0003 &  5&      50.998&            0&            0&             4.8&        269.8 &0 \\
				& 0.8 &  0.0001&   5&      61.636&            0&            0&             3&        294.4 & 0  \\
				\midrule
				500& 0.2 & 0.0003&   5&      327.896&            0&        321.6&             7.2&        392.2 &0\\
				& 0.5 &  0.0003 &    1&      3,017.650&    3.454&          3.6&            6.6&        484.2& 0 \\
				& 0.8 &  0.0001&   5&      129.684&            0&            0&          4.2&        497.4   &  0\\
				\bottomrule
			\end{tabular*}	
		\end{center}
	\end{table}
	
	\begin{table}
		\scriptsize
		\setlength{\tabcolsep}{0pt}
		\caption{\label{tab:F1B} Results of formulation $(\mFA)$ on  set B.}
		\begin{center}
			\begin{tabular*}{\textwidth}{@{\extracolsep{\fill}}rrrrrrrr}
				\toprule		
				Group&$\#opt$& $t (s)$& $\%grb-gap$  & $grb-nodes$& $paths$ & $nodes$ & $\%best-gap$  \\
				\midrule
				$G_1$ &  3&      2,246.63&    2.240&            0&           29.4&         82.2 &  2.240\\
				\bottomrule
			\end{tabular*}	
		\end{center}
	\end{table}
	
	\begin{table}
		\scriptsize
		\setlength{\tabcolsep}{0pt}
		\caption{\label{tab:F1C} Results of formulation $(\mFA)$ on  set C.}
		\begin{center}
			\begin{tabular*}{\textwidth}{@{\extracolsep{\fill}}rrrrrrrrrrr}
				\toprule		
				$n$&$p_{a}$&$p_{a'}$&$\#null-obj$ &$\#opt$&  $t (s)$& $\%grb-gap$ & $grb-nodes$  & $paths$ & $nodes$ & $\%best-gap$\\
				\midrule
				100& 0.2 & 0.2 &  0& 5&      2.197&            0&            0&                      14.0&         31.6&0 \\
				& & 0.5 & 0&   4&      803.292&   0.394&        476.4&         24.2&         51.6&  0  \\
				& & 0.8 & 0&  4&      847.054&   0.335&         25.4&          29.6&         61.4 &0  \\
				& 0.5 & 0.2 & 0& 5&      161.275&            0&            0&           22.6&         55.8  &0   \\
				& & 0.5 &0&    1&      3,262.990&    1.192&            0&           36.6&           86 & 1.183  \\
				& & 0.8 &   1& 1&      3,340.830&  7.448&            0&                 28.6&           76 & 5.586 \\
				& 0.8& 0.2 & 0& 5&      118.531&            0&                      0&         25.4&           78 &0 \\
				& & 0.5 & 0& 1&      3,373.710&     17.608&            0&            18.4&         87.4&   17.608 \\
				& & 0.8 &0&     1&      3,584.840&     27.409&            0&            13.4&         79.8  & 27.409 \\
				\midrule
				300& 0.2 & 0.2 & 2&   0&      3,672.670&       3.746&            0&                   26.2&         59.6 & 3.745\\
				\bottomrule
			\end{tabular*}	
		\end{center}
	\end{table}
	
	\begin{table}
		\scriptsize
		\setlength{\tabcolsep}{0pt}
		\caption{\label{tab:LF2A1} Results of $\overline{L\mFB}$\ on set A.1.}
		\begin{center}
			\begin{tabular*}{\textwidth}{@{\extracolsep{\fill}}rrrrrr}
				\toprule

				$n$&$p_{a}$&$p_{a'}$&$t(s)$&$paths$&$nodes$\\
				
				\midrule
				100& 0.2 & 0.2 & 0.006&            9.2&        100.0\\
				&  & 0.5 &  0.005&         9.4&        100.0\\
				& & 0.8 &   0.005&              9.4&        100.0\\
				& 0.5 & 0.2 &   0.008&              4.2&        100.0\\
				&  & 0.5 &    0.008&         4.2&        100.0\\
				& & 0.8 & 0.008&              4.4&        100.0\\
				& 0.8 & 0.2 &    0.011&            2.4&        100.0\\
				&  & 0.5 &     0.011&          2.8&        100.0\\
				& & 0.8 &  0.011&             2.6&        100.0\\
				\midrule			
				300& 0.2 & 0.2 &  0.035&   10.4&        300.0\\
				&  & 0.5 &   0.033&  10.6&        300.0\\
				& & 0.8 &   0.034&    10.8&        300.0\\
				& 0.5 & 0.2 &    0.061&          4.6&        300.0\\
				&  & 0.5 &   0.064&           5.0&        300.0\\
				& & 0.8 &   0.057&          4.8&        300.0\\
				& 0.8 & 0.2 &  0.083&           2.8&        300.0\\
				&  & 0.5 &  0.079&          3.0&        300.0\\
				& & 0.8 &   0.080&           3.0&        300.0\\
				\midrule			
				500& 0.2 & 0.2 &  0.087&       10.8&        500.0\\
				&  & 0.5 &    0.086&         10.6&        500.0\\
				& & 0.8 &   0.086&          11.2&        500.0\\
				& 0.5 & 0.2 &    0.223&          5.2&        500.0\\
				&  & 0.5 &     0.216&         5.2&        500.0\\
				& & 0.8 &   0.219&            5.0&        500.0\\
				& 0.8 & 0.2 &  0.313&            3.0&        500.0\\
				&  & 0.5 &    0.337&           3.0&        500.0\\
				& & 0.8 &   0.281&            3.0&        500.0\\
				\bottomrule
			\end{tabular*}
		\end{center}
	\end{table}
	
	\begin{table}
		\scriptsize
		\setlength{\tabcolsep}{0pt}
		\caption{\label{tab:LF2A2} Results of $\overline{L\mFB}$\ on set A.2.}
		\begin{center}
			\begin{tabular*}{\textwidth}{@{\extracolsep{\fill}}rrrrrr}
				\toprule

				$n$&$p_{a}$&$p_{a'}$&$t(s)$&$paths$&$nodes$\\
				
				\midrule
				100& 0.2 & 0.0100 &  0.005&   9.4&        100.0\\
				& 0.5 & 0.0050&  0.008&     4.2&        100.0\\
				& 0.8 & 0.0010 &    0.011&      2.8&        100.0\\
				\midrule			
				300& 0.2 & 0.0010 &   0.034&   11.0&        300.0 \\
				& 0.5 & 0.0003 &   0.069&   4.4&        300.0 \\
				& 0.8 & 0.0001 &   0.080&    2.8&        300.0\\
				\midrule			
				500& 0.2 & 0.0003 & 0.091&   11.0&        500.0 \\
				& 0.5 & 0.0003 &   0.228&     5.2&        500.0\\
				& 0.8 & 0.0001 &     0.326&      3.0&        500.0  \\
				\bottomrule
			\end{tabular*}
		\end{center}
	\end{table}
	
	\begin{table}
		\scriptsize
		\setlength{\tabcolsep}{0pt}
		\caption{\label{tab:LF2B} Results of $\overline{L\mFB}$\ on set B.}
		\begin{center}
			\begin{tabular*}{\textwidth}{@{\extracolsep{\fill}}rrrr}
				\toprule

				Group&$t(s)$&$paths$&$nodes$\\
				
				\midrule
				$G_1$  &        0.013&            28.8&         85.0  \\	
				$G_2$&      0.012&             47.4&        150.0\\ 		
				$G_3$&         0.013&            50.2&        160.0\\				
				$G_4$&         0.026&              61.8&        210.0\\
				$G_5$&        0.027&            67.0&        240.0\\			
				$G_6$&       0.034&             77.6&        270.0  \\			
				$G_7$&        0.043&              86.6&        300.0\\			
				$G_8$&        0.121&           132.0&        500.0 \\
				\bottomrule
			\end{tabular*}
		\end{center}
	\end{table}
	
	\begin{table}
		\scriptsize
		\setlength{\tabcolsep}{0pt}
		\caption{\label{tab:LF2C} Results of $\overline{L\mFB}$\ on set C.}
		\begin{center}
			\begin{tabular*}{\textwidth}{@{\extracolsep{\fill}}rrrrrr}
				\toprule

				$n$&$p_{a}$&$p_{a'}$&$t(s)$&$paths$&$nodes$\\
				
				\midrule
				100& 0.2 & 0.2 &   0.005&      67.4&        100.0 \\
				&  & 0.5 &    0.005&           66.6&        100.0 \\
				& & 0.8 &   0.005&           67.2&        100.0\\
				& 0.5 & 0.2 & 0.009&           38.2&        100.0\\
				&  & 0.5 &    0.009&           38.0&        100.0\\
				& & 0.8 &   0.009&             37.6&        100.0\\
				& 0.8 & 0.2 &    0.012&             17.2&        100.0 \\
				&  & 0.5 &  0.011&             18.6&        100.0 \\
				& & 0.8 & 0.011&            16.6&        100.0 \\
				\midrule			
				300& 0.2 & 0.2 &  0.032&    190.4&        300.0 \\
				&  & 0.5 &   0.035&            191.0&        300.0\\
				& & 0.8 &    0.035&          192.4&        300.0\\
				& 0.5 & 0.2 &  0.051&           106.0&        300.0\\
				&  & 0.5 &    0.050&         107.6&        300.0 \\
				& & 0.8 &    0.059&         107.8&        300.0 \\
				& 0.8 & 0.2 &  0.077&           46.2&        300.0\\
				&  & 0.5 &  0.076&            44.6&        300.0 \\
				& & 0.8 &   0.076&           46.2&        300.0 \\
				\midrule			
				500& 0.2 & 0.2 & 0.057&          314.4&        500.0 \\
				&  & 0.5 &   0.057&       310.6&        500.0 \\
				& & 0.8 &    0.059&         311.4&        500.0 \\
				& 0.5 & 0.2 &   0.166&         179.8&        500.0 \\
				&  & 0.5 &    0.156&         172.4&        500.0\\
				& & 0.8 &    0.155&         170.4&        500.0\\
				& 0.8 & 0.2 &  0.264&           70.2&        500.0 \\
				&  & 0.5 &  0.257&          74.2&        500.0\\
				& & 0.8 &    0.264&          69.4&        500.0 \\
				\bottomrule
			\end{tabular*}
		\end{center}
	\end{table}
	
	\begin{table}
		\scriptsize
		\setlength{\tabcolsep}{0pt}
		\caption{\label{tab:LBv1A1} Results of $\overline{L\mFB}$[+IPCs,TC-I,II] on set A.1.}
		\begin{center}
			\begin{tabular*}{\textwidth}{@{\extracolsep{\fill}}rrrrrrrrrrrr}
				\toprule
				$n$&$p_{a}$&$p_{a'}$&$\#opt$&$\#null-obj$&$lazy-cuts$&$user-cuts$&$\%grb-gap$&$t_{sep}(s)$&$t(s)$&$paths$&$nodes$\\
				
				\midrule
				
				100& 0.2 & 0.2 &     5 &          0 &       55.0 &        0.4 &          0.000 &      0.000 &      0.023 &           9.2 &      100.0 \\
				&  & 0.5 &          5 &          0 &       30.4 &        0.0 &            0.000 &      0.000 &      0.010 &            9.4 &      100.0 \\
				& & 0.8 &          5 &          0 &       13.0 &        0.0 &            0.000 &      0.000 &      0.011 &          9.4 &      100.0 \\
				& 0.5 & 0.2 & 5 &          0 &       43.6 &        0.0 &            0.000 &      0.000 &      0.017 &         4.2 &      100.0\\
				&  & 0.5 &        5 &          0 &       24.6 &        0.0 &             0.000 &      0.000 &      0.040 &         4.2 &      100.0  \\
				& & 0.8 &     5 &          0 &       11.0 &        0.0 &        0.000 &      0.000 &      0.043 &           4.4 &      100.0 \\
				& 0.8 & 0.2 &           5 &          0 &       28.4 &        0.0 &      0.000 &      0.000 &      0.055 &          2.4 &      100.0  \\
\bottomrule
			\end{tabular*}
		\end{center}
	\end{table}

	\begin{table}
		\scriptsize
		\setlength{\tabcolsep}{0pt}
		\caption{\label{tab:LBv1A2} Results of $\overline{L\mFB}$[+IPCs,TC-I,II] on set A.2.}
		\begin{center}
			\begin{tabular*}{\textwidth}{@{\extracolsep{\fill}}rrrrrrrrrrrr}
				\toprule
				$n$&$p_{a}$&$p_{a'}$&$\#opt$&$\#null-obj$&$lazy-cuts$&$user-cuts$&$\%grb-gap$&$t_{sep}(s)$&$t(s)$&$paths$&$nodes$\\
				
				\midrule
				
				100& 0.2 & 0.0100 &  0 &    0 &      103.4 &        9.6 &        790.325 &      0.000 &      0.092 &          3.8 &       37.8  \\
				& 0.5 & 0.0050 &     0 &          2 &       78.2 &       40.8 &         237.888 &      0.000 &      0.096 &         1.8 &       35.4\\
				& 0.8 & 0.0010   &  1 &          2 &       45.8 &       37.6 &      171.442 &      0.001 &      0.080 &         1.0 &       31.2  \\
				\midrule			
				300& 0.2 &  0.0010 & 0 &          2 &      167.8 &       35.2 &           514.519 &      0.000 &      0.326 &         3.4 &       75.8\\
				& 0.5 &  0.0003 &  0 &          2 &       76.6 &       54.0 &            178.265 &      0.004 &      0.614 &         2.0 &       85.4 \\
				& 0.8 & 0.0001  & 1 &          2 &       43.4 &      190.4 &            42.719 &      0.072 &      0.693 &         1.4 &      131.4 \\
				\midrule			
				500& 0.2 &0.0003  & 0 &          1 &      167.2 &       13.6 &          324.832 &      0.000 &      0.664 &         4.0 &      197.4  \\
				& 0.5 &  0.0003&  0 &          3 &       91.0 &      159.2 &           41.832 &      0.036 &      3.193 &      1.6 &      142.2\\
				& 0.8 & 0.0001 &  1 &          2 &       48.8 &      214.8 &           43.060 &      0.185 &      2.404 &       2.0 &      220.6  \\
				\bottomrule
			\end{tabular*}
		\end{center}
	\end{table}
	
	\begin{table}
		\scriptsize
		\setlength{\tabcolsep}{0pt}
		\caption{\label{tab:LBv2A2} Results of $\overline{L\mFB}$[+A-GRC] on set A.2.}
		\begin{center}
			\begin{tabular*}{\textwidth}{@{\extracolsep{\fill}}rrrrrrrrrrrr}
				\toprule
				$n$&$p_{a}$&$p_{a'}$&$\#opt$&$\#null-obj$&$lazy-cuts$&$user-cuts$&$\%grb-gap$&$t_{sep}(s)$&$t(s)$&$paths$&$nodes$\\
				
				\midrule
				
				100& 0.2 & 0.0100 &  0 &                 0 &      103.4 &        9.6 &          790.325 &      0.002 &      0.094 &         3.8 &       37.8  \\
				& 0.5 & 0.0050 &    0 &          2 &       78.2 &       40.8 &           237.888 &      0.002 &      0.098 &         1.8 &       35.4 \\
				& 0.8 & 0.0010   &   1 &          2 &       45.8 &       37.6 &          171.442 &      0.003 &      0.083 &       1.0 &       31.2 \\
				\midrule			
				300& 0.2 &  0.0010 & 0 &          2 &      167.8 &       35.2 &          514.519 &      0.027 &      0.351 &          3.4 &       75.8 \\
				& 0.5 &  0.0003 &0 &          2 &       76.6 &       54.0 &           178.265 &      0.050 &      0.657 &        2.0 &       85.4 \\
				& 0.8 & 0.0001  &    1 &          2 &       43.4 &      190.4 &           42.719 &      0.108 &      0.727 &          1.4 &      131.4 \\
				\midrule			
				500& 0.2 &0.0003  &   0 &          1 &      167.2 &       13.6 &          324.832 &      0.132 &      0.796 &    4.0 &      197.4  \\
				& 0.5 &  0.0003&  0 &          3 &       91.0 &      159.2 &           41.832 &      0.237 &      3.492 &     1.6 &      142.2   \\
				& 0.8 & 0.0001 &    1 &          2 &       48.8 &      214.8 &        43.060 &      0.378 &      2.640 &        2.0 &      220.6\\
				\bottomrule
			\end{tabular*}
		\end{center}
	\end{table}
	
	\begin{table}
		\scriptsize
		\setlength{\tabcolsep}{0pt}
		\caption{\label{tab:LBv2CA2} Results of $\overline{L\mFB}$[+A-GRC] (with Cliquer) on set A.2.}
		\begin{center}
			\begin{tabular*}{\textwidth}{@{\extracolsep{\fill}}rrrrrrrrrrrr}
				\toprule
				$n$&$p_{a}$&$p_{a'}$&$\#opt$&$\#null-obj$&$lazy-cuts$&$user-cuts$&$\%grb-gap$&$t_{sep}(s)$&$t(s)$&$paths$&$nodes$\\
				
				\midrule
				
				100& 0.2 & 0.0100 &    0 &        0 &      103.4 &       12.0 &        790.325 &      0.132 &      0.227 &         3.8 &       37.8    \\
				& 0.5 & 0.0050 &   0 &          2 &       78.2 &       42.4 &          237.888 &      0.053 &      0.151 &       1.8 &       35.4 \\
				& 0.8 & 0.0010   &     1 &          2 &       45.8 &       37.6 &          171.442 &      0.056 &      0.141 &      1.0 &       31.2  \\
				\midrule			
				300& 0.2 &  0.0010 &  0 &          2 &      167.8 &       46.0 &           514.519 &      0.118 &      0.442 &3.4 &       75.8  \\
				& 0.5 &  0.0003 &   0 &          2 &       76.6 &       58.8 &            178.265 &      0.149 &      0.753 &         2.0 &       85.4\\
				& 0.8 & 0.0001  &   1 &          2 &       43.4 &      190.4 &             42.719 &      0.162 &      0.781 &           1.4 &      131.4  \\
				\midrule			
				500& 0.2 &0.0003  &  0 &          1 &      167.2 &      139.4 &          324.832 &      0.264 &      0.929 &        4.0 &      197.4 \\
				& 0.5 &  0.0003&   0 &          3 &       91.0 &      159.2 &          41.832 &      0.316 &      3.530 &         1.6 &      142.2  \\
				& 0.8 & 0.0001 & 1 &          2 &       48.8 &      214.8 &         43.060 &      0.486 &      2.685 &    2.0 &      220.6 \\
				\bottomrule
			\end{tabular*}
		\end{center}
	\end{table}
	
	\begin{table}
		\scriptsize
		\setlength{\tabcolsep}{0pt}
		\caption{\label{tab:LBv3A2} Results of $\overline{L\mFB}$[+RC,$\pm,-,+$] on set A.2.}
		\begin{center}
			\begin{tabular*}{\textwidth}{@{\extracolsep{\fill}}rrrrrrrrrrrr}
				\toprule
				$n$&$p_{a}$&$p_{a'}$&$\#opt$&$\#null-obj$&$lazy-cuts$&$user-cuts$&$\%grb-gap$&$t_{sep}(s)$&$t(s)$&$paths$&$nodes$\\
				
				\midrule
				
				100& 0.2 & 0.0100 &   0 &          1 &      107.6 &       11.2 &            188.969 &      0.004 &      0.099 &    4.2 &       43.6   \\
				& 0.5 & 0.0050 &    0 &          2 &       83.0 &       40.6 &           237.888 &      0.003 &      0.110 &           1.8 &       35.4  \\
				& 0.8 & 0.0010   &    1 &          0 &       49.0 &       46.4 &            138.150 &      0.005 &      0.088 &         2.2 &       51.8   \\
				\midrule			
				300& 0.2 &  0.0010 &   0 &          2 &      165.6 &       34.2 &            515.293 &      0.047 &      0.329 &         3.4 &       75.2\\
				& 0.5 &  0.0003 &  0 &          2 &       79.0 &       59.0 &            178.265 &      0.084 &      0.730 &        2.0 &       85.4  \\
				& 0.8 & 0.0001  &   1 &          2 &       44.6 &      221.8 &             42.719 &      0.182 &      0.765 &         1.4 &      131.4 \\
				\midrule			
				500& 0.2 &0.0003  &  0 &          1 &      166.2 &       14.0 &         55.812 &      0.225 &      0.879 &       4.6 &      221.0\\
				& 0.5 &  0.0003&    0 &          3 &       91.0 &      159.2 &        41.832 &      0.415 &      3.641 &     1.6 &      142.2  \\
				& 0.8 & 0.0001 &  1 &          2 &       48.6 &      338.2 &         43.235 &      0.816 &      3.528 &    2.0 &      220.4 \\
				\bottomrule
			\end{tabular*}
		\end{center}
	\end{table}
	
	\begin{table}
		\scriptsize
		\setlength{\tabcolsep}{0pt}
		\caption{\label{tab:LBv3CA2} Results of $\overline{L\mFB}$[+RC,$\pm,-,+$] (with Cliquer) on set A.2.}
		\begin{center}
			\begin{tabular*}{\textwidth}{@{\extracolsep{\fill}}rrrrrrrrrrrr}
				\toprule
				$n$&$p_{a}$&$p_{a'}$&$\#opt$&$\#null-obj$&$lazy-cuts$&$user-cuts$&$\%grb-gap$&$t_{sep}(s)$&$t(s)$&$paths$&$nodes$\\
				
				\midrule
				
				100& 0.2 & 0.0100 &    0 &                   0 &      105.2 &       18.0 &          789.789 &      0.142 &      0.241 &        3.8 &       38.0 \\
				& 0.5 & 0.0050 &    0 &          2 &       78.2 &       42.4 &           237.888 &      0.047 &      0.147 &         1.8 &       35.4  \\
				& 0.8 & 0.0010   &    1 &          0 &       57.4 &       71.4 &          120.430 &      0.093 &      0.205 &         2.2 &       60.4 \\
				\midrule			
				300& 0.2 &  0.0010 &   0 &          2 &      167.8 &       46.0 &           514.519 &      0.114 &      0.443 &          3.4 &       75.8 \\
				& 0.5 &  0.0003 &    0 &          2 &       77.2 &       71.6 &            178.265 &      0.179 &      0.803 &          2.0 &       85.4\\
				& 0.8 & 0.0001  &     1 &          2 &       43.4 &      190.4 &          42.719 &      0.199 &      0.822 &          1.4 &      131.4\\
				\midrule			
				500& 0.2 &0.0003  &  0 &          1 &      167.2 &      139.4 &         324.832 &      0.300 &      0.974 &       4.0 &      197.4  \\
				& 0.5 &  0.0003&    0 &          3 &       91.0 &      159.2 &             41.832 &      0.499 &      3.729 &         1.6 &      142.2  \\
				& 0.8 & 0.0001 &  1 &          2 &       48.4 &      278.0 &        5.997 &      0.822 &      3.361 &        2.0 &      217.4\\
				\bottomrule
			\end{tabular*}
		\end{center}
	\end{table}
	
	\begin{table}
		\scriptsize
		\setlength{\tabcolsep}{0pt}
		\caption{\label{tab:LBv1B} Results of $\overline{L\mFB}$[+IPCs,TC-I,II] on set B.}
		\begin{center}
			\begin{tabular*}{\textwidth}{@{\extracolsep{\fill}}rrrrrrrrrr}
				\toprule

				Group&$\#opt$&$\#null-obj$&$lazy-cuts$&$user-cuts$&$\%grb-gap$&$t_{sep}(s)$&$t(s)$&$paths$&$nodes$\\
				
				\midrule
				$G_1$  &  0 &          3 &       74.6 &        8.0 &     117.820 &      0.000 &      0.223 &          7.2 &       21.4 \\	
				$G_2$&   0 &          2 &      159.8 &        8.0 &           16.623 &      0.001 &      0.593 &        25.8 &       77.2\\ 		
				$G_3$&     0 &          2 &      239.4 &        8.0 &         16.991 &      0.001 &      0.907 &       27.2 &       80.6   \\				
				$G_4$&    0 &          4 &      194.0 &        6.0 &      16.665 &      0.004 &      1.608 &         11.4 &       36.0 \\
				$G_5$&     0 &          4 &      291.0 &        8.0 &          23.713 &      0.003 &      1.675 &      11.8 &       38.8 \\			
				$G_6$&   0 &          2 &      261.0 &        8.4 &     25.948 &      0.004 &      1.412 &         38.4 &      128.8\\			
				$G_7$&     0 &          4 &      434.0 &        8.8 &       22.955 &      0.009 &      3.107 &          13.8 &       48.8\\			
				$G_8$&             0 &          1 &       99.0 &        6.0 &             -&      0.014 &      4.754 &       0.0 &        0.0  \\
				\bottomrule
			\end{tabular*}
		\end{center}
	\end{table}
	
	\begin{table}
		\scriptsize
		\setlength{\tabcolsep}{0pt}
		\caption{\label{tab:LBv1C} Results of $\overline{L\mFB}$[+IPCs,TC-I,II] on set C.}
		\begin{center}
			\begin{tabular*}{\textwidth}{@{\extracolsep{\fill}}rrrrrrrrrrrr}
				\toprule
				$n$&$p_{a}$&$p_{a'}$&$\#opt$&$\#null-obj$&$lazy-cuts$&$user-cuts$&$\%grb-gap$&$t_{sep}(s)$&$t(s)$&$paths$&$nodes$\\
				
				\midrule
				
				100& 0.2 & 0.2 &     5 &          0 &       84.2 &        0.8 &            0.000 &      0.000 &      0.101 &        14.0 &       31.6  \\
				&  & 0.5 &   1 &          0 &       89.2 &        6.8 &         29.844 &      0.000 &      0.165 &         20.0 &       43.8  \\
				& & 0.8 &   2 &          0 &       91.0 &        8.0 &        3.128 &      0.000 &      0.187 &        28.0 &       59.8 \\
				& 0.5 & 0.2 &    4 &          1 &       64.0 &        4.8 &       0.000 &      0.000 &      0.130 &        18.4 &       45.2 \\
				&  & 0.5 &    0 &          2 &      128.6 &       10.4 &   90.677 &      0.000 &      0.140 &         14.2 &       35.2  \\
				& & 0.8 &  1 &          4 &       93.6 &        6.8 &       0.000 &      0.000 &      0.138 &           7.8 &       20.0 \\
				& 0.8 & 0.2 &   5 &          0 &       37.6 &        0.0 &       0.000 &      0.000 &      0.081 &         25.4 &       78.0  \\
\bottomrule
			\end{tabular*}
		\end{center}
	\end{table}
	
	\begin{table}
		\scriptsize
		\setlength{\tabcolsep}{0pt}
		\caption{\label{tab:LBv2C} Results of $\overline{L\mFB}$[+A-GRC] on set C.}
		\begin{center}
			\begin{tabular*}{\textwidth}{@{\extracolsep{\fill}}rrrrrrrrrrrr}
				\toprule
				$n$&$p_{a}$&$p_{a'}$&$\#opt$&$\#null-obj$&$lazy-cuts$&$user-cuts$&$\%grb-gap$&$t_{sep}(s)$&$t(s)$&$paths$&$nodes$\\
				
				\midrule
				
				100& 0.2 & 0.2 &        5 &          0 &       84.2 &        0.8 &              0.000 &      0.003 &      0.112 &        14.0 &       31.6  \\
				&  & 0.5 &        1 &          0 &       89.2 &        6.8 &      29.844 &      0.004 &      0.188 &        20.0 &       43.8     \\
				& & 0.8 &     2 &          0 &       91.0 &        8.0 &    3.128 &      0.005 &      0.184 &         28.0 &       59.8     \\
				& 0.5 & 0.2 &   4 &          1 &       64.0 &        4.8 &      0.000 &      0.004 &      0.215 &      18.4 &       45.2  \\
				&  & 0.5 &     0 &          2 &      128.6 &       10.4 &   90.677 &      0.000 &      0.140 &         14.2 &       35.2  \\
				& & 0.8 &  1 &          4 &       93.6 &        6.8 &     0.000 &      0.001 &      0.094 &         7.8 &       20.0  \\
				& 0.8 & 0.2 &    5 &          0 &       37.6 &        0.0 &       0.000 &      0.000 &      0.026 &         25.4 &       78.0  \\
\bottomrule
			\end{tabular*}
		\end{center}
	\end{table}
	
	\begin{table}
		\scriptsize
		\setlength{\tabcolsep}{0pt}
		\caption{\label{tab:LBv2cC} Results of $\overline{L\mFB}$[+A-GRC] (with cliquer) on set C.}
		\begin{center}
			\begin{tabular*}{\textwidth}{@{\extracolsep{\fill}}rrrrrrrrrrrr}
				\toprule
				$n$&$p_{a}$&$p_{a'}$&$\#opt$&$\#null-obj$&$lazy-cuts$&$user-cuts$&$\%grb-gap$&$t_{sep}(s)$&$t(s)$&$paths$&$nodes$\\
				
				\midrule
				
				100& 0.2 & 0.2 & 5 &          0 &       84.2 &        0.8 &          0.000 &      0.101 &      0.133 &           14.0 &       31.6  \\
				&  & 0.5 &      1 &          0 &       89.2 &        6.8 &    29.844 &      0.053 &      0.105 &        20.0 &       43.8    \\
				& & 0.8 &     2 &          0 &       91.0 &        8.2 &       3.128 &      0.062 &      0.115 &          28.0 &       59.8       \\
				& 0.5 & 0.2 &    4 &          1 &       64.0 &        4.8 &      0.000 &      0.097 &      0.288 &        18.4 &       45.2  \\
				&  & 0.5 &     0 &          2 &      128.6 &       10.4 &          90.677 &      0.028 &      0.357 &         14.2 &       35.2 \\
				& & 0.8 &   0 &          4 &       96.6 &        8.8 &          1.004 &      0.035 &      0.134 &          7.6 &       19.8 \\
				& 0.8 & 0.2 &     5 &          0 &       37.6 &        0.0 &        0.000 &      0.000 &      0.021 &       25.4 &       78.0  \\
\bottomrule
			\end{tabular*}
		\end{center}
	\end{table}
	
	\begin{table}
		\scriptsize
		\setlength{\tabcolsep}{0pt}
		\caption{\label{tab:LBv3C} Results of $\overline{L\mFB}$[+RC,$\pm,-,+$] on set C.}
		\begin{center}
			\begin{tabular*}{\textwidth}{@{\extracolsep{\fill}}rrrrrrrrrrrr}
				\toprule
				$n$&$p_{a}$&$p_{a'}$&$\#opt$&$\#null-obj$&$lazy-cuts$&$user-cuts$&$\%grb-gap$&$t_{sep}(s)$&$t(s)$&$paths$&$nodes$\\
				
				\midrule
				
				100& 0.2 & 0.2 &   5 &          0 &       76.6 &        3.2 &      0.000 &      0.002 &      0.024 &       14.0 &       31.6   \\
				&  & 0.5 &        2 &          0 &       92.4 &        9.6 &           2.692 &      0.003 &      0.058 &         23.4 &       50.6     \\
				& & 0.8 &        4 &          0 &       91.4 &       11.4 &       1.105 &      0.005 &      0.062 &         29.0 &       60.8    \\
				& 0.5 & 0.2 &     4 &          0 &       66.6 &        6.6 &          2.696 &      0.002 &      0.062 &        22.2 &       54.4 \\
				&  & 0.5 &   0 &          2 &      127.6 &       10.6 &          90.677 &      0.001 &      0.140 &        14.2 &       35.2   \\
				& & 0.8 &    0 &          5 &       98.8 &        9.2 &   - &      0.005 &      0.292 &              0.0 &        0.0 \\
				& 0.8 & 0.2 &  5 &          0 &       37.6 &        0.0 &         0.000 &      0.000 &      0.044 &          25.4 &       78.0   \\
\bottomrule
			\end{tabular*}
		\end{center}
	\end{table}
	
	\begin{table}
		\scriptsize
		\setlength{\tabcolsep}{0pt}
		\caption{\label{tab:LBv3cC} Results of $\overline{L\mFB}$[+RC,$\pm,-,+$] (with cliquer) on set C.}
		\begin{center}
			\begin{tabular*}{\textwidth}{@{\extracolsep{\fill}}rrrrrrrrrrrr}
				\toprule
				$n$&$p_{a}$&$p_{a'}$&$\#opt$&$\#null-obj$&$lazy-cuts$&$user-cuts$&$\%grb-gap$&$t_{sep}(s)$&$t(s)$&$paths$&$nodes$\\
				
				\midrule
				
				100& 0.2 & 0.2 &    5 &          0 &       81.8 &        4.8  &      0.000 &      0.146 &      0.241 &          14.0 &       31.6 \\
				&  & 0.5 &      2 &          0 &       95.2 &       10.4 &        4.018 &      0.127 &      0.312 &          23.0 &       50.0      \\
				& & 0.8 &      3 &          1 &       90.8 &       12.6 &            1.243 &      0.154 &      0.326 &          23.6 &       49.4      \\
				& 0.5 & 0.2 &     3 &          0 &       72.8 &        9.6 &        3.479 &      0.126 &      0.348 &          22.0 &       54.0  \\
				&  & 0.5 &   0 &          2 &      128.6 &       10.4 &        90.677 &      0.040 &      0.481 &         14.2 &       35.2  \\
				& & 0.8 &     0 &          4 &       96.6 &        8.8 &             1.004 &      0.092 &      0.378 &           7.6 &       19.8  \\
				& 0.8 & 0.2 &    5 &          0 &       37.6 &        0.0 &        0.000 &      0.000 &      0.054 &        25.4 &       78.0  \\
\bottomrule
			\end{tabular*}
		\end{center}
	\end{table}
	
	\begin{table}
		\scriptsize
		\setlength{\tabcolsep}{2.5pt}
		\caption{\label{tab:BCv1A2} Results of $\overline{\mFB}$[+IPCs,TC-I,II] on A.2.}
		\begin{center}
		\resizebox{\textwidth}{!}{\begin{tabular}{rrrrrrrrrrrrrrr}
				\toprule

				$n$&$p_{a}$&$p_{a'}$&$\#opt$&$\#opt_0$&$\#null-obj$&$lazy-cuts$&$user-cuts$&$grb-nodes$&$\%grb-gap$&$t_{sep}(s)$&$t(s)$&$paths$&$nodes$&$\%best-gap$\\
				
				\midrule
				100& 0.2 & 0.0100 &   0 &          0 &          0 &    37511.6 &  1349274.8 &   833286.0 &      6.686 &      0.095 &   3600.228 &          7.4 &       82.2  &2.626  \\
				& 0.5 & 0.0050 &    0 &          0 &          0 &    85084.8 &  3345849.6 &   203642.2 &     15.850 &      8.096 &   3600.576 &           4.8 &       87.6  &13.363  \\
				& 0.8 & 0.0010   &    1 &          1 &          2 &    32467.2 &  2446184.0 &    91004.4 &    171.442 &    172.834 &   2880.984 &         1.0 &       31.2 &164.321  \\
				\midrule			
				300& 0.2 &  0.0010 &   0 &          0 &          0 &    57500.8 &  1139135.6 &   182897.6 &    258.600 &      0.245 &   3600.949 &   6.0 &      160.4 &250.601 \\
				& 0.5 &  0.0003 &   0 &          0 &          0 &    35742.8 &  3777502.4 &    83571.4 &     48.306 &    423.527 &   3603.574 &         3.8 &      196.4  &47.329 \\
				& 0.8 & 0.0001  &      2 &          1 &          1 &    12248.2 &  3001847.6 &    23518.4 &     31.835 &   1501.563 &   2811.390 &         2.2 &      191.4 &31.835   \\
				\midrule			
				500& 0.2 &0.0003  &    0 &          0 &          0 &    35589.6 &  1065888.0 &   112918.6 &     43.409 &     21.092 &   3600.832 &        6.6 &      312.0 &34.231\\
				& 0.5 &  0.0003&  0 &          0 &          0 &    20302.4 &  4354274.8 &    41041.2 &   1301.135 &    902.610 &   3610.914 &        2.8 &      181.2  &1278.362   \\
				& 0.8 & 0.0001 &    3 &          1 &          1 &     2435.2 &  1454322.0 &     5052.4 &     87.612 &   1325.932 &   2025.141 &         2.6 &      320.6&87.612 \\
				\bottomrule
			\end{tabular}}
		\end{center}
	\end{table}
	
	\begin{table}
		\scriptsize
		\setlength{\tabcolsep}{2.5pt}
		\caption{\label{tab:BCv2A2} Results of $\overline{\mFB}$[+A-GRC] on A.2.}
		\begin{center}
		\resizebox{\textwidth}{!}{	\begin{tabular}{rrrrrrrrrrrrrrr}
				\toprule

				$n$&$p_{a}$&$p_{a'}$&$\#opt$&$\#opt_0$&$\#null-obj$&$lazy-cuts$&$user-cuts$&$grb-nodes$&$\%grb-gap$&$t_{sep}(s)$&$t(s)$&$paths$&$nodes$&$\%best-gap$\\
				
				\midrule
				100& 0.2 & 0.0100 &   0 &          0 &          0 &    37461.2 &  1347020.8 &   830746.4 &      6.686 &      0.047 &   3600.153 &          7.4 &       82.2  & 2.626   \\
				& 0.5 & 0.0050 &    0 &          0 &          0 &    85023.6 &  3342514.4 &   203468.4 &     15.850 &      8.058 &   3600.451 &           4.8 &       87.6  &13.363 \\
				& 0.8 & 0.0010   &     1 &          1 &          2 &    32199.2 &  2423312.4 &    90116.4 &    171.442 &    172.167 &   2880.198 &     1.0 &       31.2 &164.321    \\
				\midrule			
				300& 0.2 &  0.0010 &  0 &          0 &          0 &    57639.2 &  1143156.8 &   183424.8 &    258.600 &      0.191 &   3600.810 &      6.0 &      160.4 & 250.601   \\
				& 0.5 &  0.0003 &  0 &          0 &          0 &    35577.0 &  3774384.8 &    83247.2 &     48.306 &    418.150 &   3605.121 &        3.8 &      196.4  &47.329  \\
				& 0.8 & 0.0001  &    2 &          1 &          1 &    12286.6 &  3010983.6 &    23584.6 &     31.835 &   1499.797 &   2811.593 &       2.2 &      191.4 &31.835  \\
				\midrule			
				500& 0.2 &0.0003  &    0 &          0 &          0 &    35487.4 &  1062688.0 &   112599.2 &     43.409 &     21.352 &   3600.759 &       6.6 &      312.0& 34.231\\
				& 0.5 &  0.0003&     0 &          0 &          0 &    20334.2 &  4363739.6 &    41109.4 &   1301.135 &    898.968 &   3609.450 &         2.8 &      181.2   &1278.362  \\
				& 0.8 & 0.0001 &     3 &          1 &          1 &     2428.0 &  1449458.8 &     5036.6 &     87.612 &   1333.041 &   2031.632 &        2.6 &      320.6 &87.612\\
				\bottomrule
			\end{tabular}}
		\end{center}
	\end{table}
	\begin{table}
		\scriptsize
		\setlength{\tabcolsep}{2.5pt}
		\caption{\label{tab:BCv2cA2} Results of $\overline{\mFB}$[+A-GRC] (with cliquer) on A.2.}
		\begin{center}
		\resizebox{\textwidth}{!}{	\begin{tabular}{rrrrrrrrrrrrrrr}
				\toprule

				$n$&$p_{a}$&$p_{a'}$&$\#opt$&$\#opt_0$&$\#null-obj$&$lazy-cuts$&$user-cuts$&$grb-nodes$&$\%grb-gap$&$t_{sep}(s)$&$t(s)$&$paths$&$nodes$&$\%best-gap$\\
				
				\midrule
				100& 0.2 & 0.0100 & 0 &  0&        0 &    15221.4 &   311469.6 &   115354.2 &      7.734 &   3229.301 &   3629.678 &          7.4 &       82.0 &3.020     \\
				& 0.5 & 0.0050 &   0 &          0 &          0 &    49969.8 &  2043580.2 &   112966.6 &     23.389 &   2380.769 &   3610.800 &         4.6 &       84.4 &20.506  \\
				& 0.8 & 0.0010   &       1 &          1 &          2 &    22052.2 &  1627237.6 &    57124.0 &    171.442 &   1683.729 &   2888.305 &        1.0 &       31.2 &164.321\\
				\midrule			
				300& 0.2 &  0.0010 &   0 &          0 &          0 &    37332.6 &   995313.2 &   132524.6 &     20.135 &   1702.811 &   3605.502 &       8.0 &      225.8 &18.842 \\
				& 0.5 &  0.0003 &      0 &          0 &          0 &    28284.2 &  2810637.4 &    64627.2 &     49.439 &   1573.005 &   3619.519 &          3.8 &      193.8 &48.458 \\
				& 0.8 & 0.0001  &     1 &          1 &          1 &    10992.8 &  2718945.2 &    21211.4 &     33.248 &   1674.493 &   2909.805 &         2.2 &      188.2 & 33.248 \\
				\midrule			
				500& 0.2 &0.0003  &     0 &          0 &          0 &    27219.4 &  1384936.4 &    89874.8 &     36.742 &   1012.045 &   3602.722 &  6.6 &      317.2 &28.403 \\
				& 0.5 &  0.0003&       0 &          0 &          0 &    17968.4 &  3818246.0 &    36137.0 &   1301.135 &   1240.582 &   3624.746 &          2.8 &      181.2&1278.362  \\
				& 0.8 & 0.0001 &    3 &          1 &          1 &     2428.0 &  1449458.8 &     5036.6 &     87.612 &   1402.949 &   2093.338 &        2.6 &      320.6& 87.612\\
				\bottomrule
			\end{tabular}}
		\end{center}
	\end{table}
	\begin{table}
		\scriptsize
		\setlength{\tabcolsep}{2.5pt}
		\caption{\label{tab:BCv3A2} Results of $\overline{\mFB}$[+RC,$\pm,-,+$] on A.2.}
		\begin{center}
		\resizebox{\textwidth}{!}{	\begin{tabular}{rrrrrrrrrrrrrrr}
				\toprule

				$n$&$p_{a}$&$p_{a'}$&$\#opt$&$\#opt_0$&$\#null-obj$&$lazy-cuts$&$user-cuts$&$grb-nodes$&$\%grb-gap$&$t_{sep}(s)$&$t(s)$&$paths$&$nodes$&$\%best-gap$\\
				
				\midrule
				100& 0.2 & 0.0100 &   2 &          0 &          0 &    24033.4 &   622328.4 &   252217.0 &      2.858 &     52.533 &   2171.037 &          7.6 &       83.4 &0.431  \\
				& 0.5 & 0.0050 &  0 &          0 &          1 &    79742.4 &  3080873.0 &   193089.2 &      6.505 &      7.326 &   3600.703 &         4.0 &       75.0 & 4.608  \\
				& 0.8 & 0.0010   &      3 &          1 &          0 &    17070.2 &  1417658.6 &    41541.2 &    102.865 &     73.617 &   1459.996 &          2.4 &       69.0  &98.592  \\
				\midrule			
				300& 0.2 &  0.0010 &    1 &          0 &          0 &    41369.8 &  1082301.4 &   140219.4 &     57.033 &    132.907 &   2881.114 &          7.2 &      203.0& 55.867 \\
				& 0.5 &  0.0003 &    1 &          0 &          0 &    31379.2 &  3607954.4 &    73126.2 &     51.151 &    454.872 &   2888.236 &        4.4 &      194.6  &50.280 \\
				& 0.8 & 0.0001  &    3 &          1 &          0 &    10090.6 &  2631569.0 &    19376.4 &     12.310 &    960.649 &   2086.653 &         3.0 &      271.6 &11.674  \\
				\midrule			
				500& 0.2 &0.0003  &    0 &          0 &          0 &    35095.0 &  1062610.6 &   108667.4 &     44.649 &    284.691 &   3603.035 &         6.6 &      311.2 & 35.316\\
				& 0.5 &  0.0003&     0 &          0 &          0 &    19892.4 &  4266088.8 &    40196.6 &   1301.135 &    982.717 &   3619.633 &       2.8 &      181.2 &1278.362   \\
				& 0.8 & 0.0001 &  2 &          1 &          1 &     3294.0 &  1666173.4 &     6628.4 &     85.235 &   1759.262 &   2562.877 &       2.6 &      287.4 &84.663 \\
				\bottomrule
			\end{tabular}}
		\end{center}
	\end{table}
	\begin{table}
		\tiny
		\setlength{\tabcolsep}{2.5pt}
		\caption{\label{tab:BCv3cA2} Results of $\overline{\mFB}$[+RC,$\pm,-,+$] (with Cliquer) on A.2.}
		\begin{center}
			\resizebox{\textwidth}{!}{
			\begin{tabular}{rrrrrrrrrrrrrrr}
				\toprule

				{\fontsize{5.5}{7}$n$}&{\fontsize{5.5}{7}$p_{a}$}&{\fontsize{5.5}{7}$p_{a'}$}&{\fontsize{5.5}{7}$\#opt$}&{\fontsize{5.5}{7}$\#opt_0$}&{\fontsize{5.5}{7}$\#null-obj$}&{\fontsize{5.5}{7}$lazy-cuts$}&{\fontsize{5.5}{7}$user-cuts$}&{\fontsize{5.5}{7}$grb-nodes$}&{\fontsize{5.5}{7}$\%grb-gap$}&{\fontsize{5.5}{7}$t_{sep}(s)$}&{\fontsize{5.5}{7}$t(s)$}&{\fontsize{5.5}{7}$paths$}&{\fontsize{5.5}{7}$nodes$}&{\fontsize{5.5}{7}$\%best-gap$}\\
				
				\midrule
				100& 0.2 & 0.0100 &    1 &          0 &          0 &    15160.2 &   318223.2 &   111377.8 &      3.224 &   3126.714 &   3506.354 &         7.6 &       83.4  &0.534 \\
				& 0.5 & 0.0050 &    0 &          0 &          0 &    50091.6 &  2044488.4 &   113229.0 &     23.389 &   2365.427 &   3606.232 &         4.6 &       84.4 &20.506 \\
				& 0.8 & 0.0010   &     3 &          1 &          0 &    15942.8 &  1273382.4 &    36841.2 &    102.865 &    856.471 &   1693.634 &        2.4 &       69.0 &98.592  \\
				\midrule			
				300& 0.2 &  0.0010 &    0 &          0 &          0 &    36847.2 &   987080.8 &   131029.6 &     20.135 &   1744.713 &   3607.723 &         8.0 &      225.8 &18.842 \\
				& 0.5 &  0.0003 &     0 &          0 &          0 &    27827.2 &  2806268.0 &    64255.8 &     48.224 &   1692.417 &   3619.494 &           4.0 &      196.6 &47.247  \\
				& 0.8 & 0.0001  &   1 &          1 &          1 &    11074.6 &  2731446.8 &    21375.0 &     33.248 &   1690.149 &   2919.213 &       2.2 &      188.2 & 33.248 \\
				\midrule			
				500& 0.2 &0.0003  &   0 &          0 &          0 &    26029.6 &  1325154.6 &    84348.6 &     40.849 &   1240.148 &   3613.720 &        6.6 &      313.8 & 31.994 \\
\bottomrule
			\end{tabular}}
		\end{center}
	\end{table}
	
	\begin{table}
		\tiny
		\setlength{\tabcolsep}{2.5pt}
		\caption{\label{tab:compareLBA2} Comparison of lower bounds on set A.2.}
		\begin{center}
			\resizebox{\textwidth}{!}{
			\begin{tabular}{rrrrrrrrrrrrrrrrrrrrr}
				
				\toprule
				$n$& $p_a$& $p_{a'}$& \multicolumn{3}{c}{$\overline{L\mFB}$} & \multicolumn{3}{c}{$\overline{L\mFB}$[+IPCs,TC-I,II]} & \multicolumn{3}{c}{$\overline{L\mFB}$[+A-GRC]} &  \multicolumn{3}{c}{$\overline{L\mFB}$[+A-GRC](C)} & \multicolumn{3}{c}{$\overline{L\mFB}$[+RC,$\pm,-,+$]} & \multicolumn{3}{c}{$\overline{L\mFB}$[+RC,$\pm,-,+$](C)}\\
				\cmidrule{4-6}	\cmidrule{7-9} \cmidrule{10-12} \cmidrule{13-15} \cmidrule{16-18} \cmidrule{19-21} 			
				&&&{\fontsize{6.5}{8}$\#bs$}&{\fontsize{6.5}{8} $\%obj$}&{\fontsize{5.5}{7} $\%best-gap$}&{\fontsize{6.5}{8}$\#bs$}&{\fontsize{6.5}{8} $\%obj$}&{\fontsize{5.5}{7} $\%best-gap$}&{\fontsize{6.5}{8}$\#bs$}&{\fontsize{6.5}{8} $\%obj$}&{\fontsize{5.5}{7} $\%best-gap$}&{\fontsize{6.5}{8}$\#bs$}&{\fontsize{6.5}{8} $\%obj$}&{\fontsize{5.5}{7} $\%best-gap$}&{\fontsize{6.5}{8}$\#bs$}&{\fontsize{6.5}{8} $\%obj$}&{\fontsize{5.5}{7} $\%best-gap$}&{\fontsize{6.5}{8}$\#bs$}&{\fontsize{6.5}{8} $\%obj$}&{\fontsize{5.5}{7} $\%best-gap$}\\
				\midrule
				100& 0.2 & 0.0100 &0 & 16.197 & 16.196 & 0 & 121.715 & 775.743 & 0 & 121.715 & 775.743 & 0 & 121.715 & 775.743 & 0 & 92.213 & 181.741 & 0 & 120.547 & 775.253\\
				& 0.5 & 0.0050&0 & 3.209 & 3.209 & 0 & 173.444 & 237.469 & 0 & 173.444 & 237.469 & 0 & 173.444 & 237.469 & 0 & 173.444 & 237.469 & 0 & 173.444 & 237.469\\
				& 0.8 & 0.0010 & 1 & 3.401 & 3.401 & 1 & 209.625 & 164.321 & 1 & 209.625 & 164.321 & 1 & 209.625 & 164.321 & 1 & 86.512 & 132.319 & 1 & 59.946 & 114.961\\          
				\midrule
				
				300& 0.2 &  0.0010 &0 & 12.999 & 12.999  & 0 & 244.340 & 496.499 & 0 & 244.340 & 496.499 & 0 & 244.340 & 496.499 & 0 & 247.088 & 497.269 & 0 & 244.340 & 496.499\\
				& 0.5 &  0.0003 &0 & 10.068 & 10.068 & 0 & 215.931 & 176.370 & 0 & 215.931 & 176.370 & 0 & 215.931 & 176.370 & 0 & 215.931 & 176.370 & 0 & 215.931 & 176.370 \\
\bottomrule
			\end{tabular}}
		\end{center}
	\end{table}
	
	\begin{table}
		\scriptsize
		\setlength{\tabcolsep}{2.5pt}
		\caption{\label{tab:compareBCA2} Comparison of BC variants on set A.2.}
		\begin{center}
			\resizebox{\textwidth}{!}{
			\begin{tabular}{rrrrrrrrrrrrrrrrrr}
				
				\toprule
				$n$   & $p_a$& $p_{a'}$&\multicolumn{3}{c}{$\overline{\mFB}$[+IPCs,TC-I,II]} & \multicolumn{3}{c}{$\overline{\mFB}$[+A-GRC]} &  \multicolumn{3}{c}{$\overline{\mFB}$[+A-GRC](C)} & \multicolumn{3}{c}{$\overline{\mFB}$[+RC,$\pm,-,+$]} & \multicolumn{3}{c}{$\overline{\mFB}$[+RC,$\pm,-,+$](C)}\\
				\cmidrule{4-6}	\cmidrule{7-9} \cmidrule{10-12} \cmidrule{13-15} \cmidrule{16-18} 			
				&&&{\fontsize{6.5}{8}$\#bs$}&{\fontsize{6.5}{8} $\%obj$}&{\fontsize{5.5}{7} $\%best-gap$}&{\fontsize{6.5}{8}$\#bs$}&{\fontsize{6.5}{8} $\%obj$}&{\fontsize{5.5}{7} $\%best-gap$}&{\fontsize{6.5}{8}$\#bs$}&{\fontsize{6.5}{8} $\%obj$}&{\fontsize{5.5}{7} $\%best-gap$}&{\fontsize{6.5}{8}$\#bs$}&{\fontsize{6.5}{8} $\%obj$}&{\fontsize{5.5}{7} $\%best-gap$}&{\fontsize{6.5}{8}$\#bs$}&{\fontsize{6.5}{8} $\%obj$}&{\fontsize{5.5}{7} $\%best-gap$}\\
				
				\midrule
				100& 0.2 & 0.0100 &  3 & 1.946 & 2.626 & 3 & 1.946 & 2.626 & 3 & 2.195 & 3.020 & 4 & 0.480 & 0.431 & 4 & 0.480 & 0.534\\
				& 0.5 & 0.0050 &3 & 10.506 & 13.363 & 3 & 10.506 & 13.363 & 3 & 14.695 & 20.506 & 3 & 29.069 & 4.608 & 3 & 14.695 & 20.506\\
				& 0.8 & 0.0010&1 & 209.625 & 164.321 & 1 & 209.625 & 164.321 & 1 & 209.625 & 164.321 & 3 & 40.008 & 98.592 & 3 & 40.008 & 98.592 \\
				
				\midrule			
				300& 0.2 &  0.0010&2 & 62.720 & 250.601 & 2 & 62.720 & 250.601 & 2 & 15.590 & 18.842 & 2 & 28.572 & 55.867 & 2 & 15.590 & 18.842\\
				& 0.5 &  0.0003 & 1 & 37.373 & 47.329 & 1 & 37.373 & 47.329 & 1 & 39.217 & 48.458 & 2 & 38.646 & 50.280 & 1 & 37.234 & 47.247\\
				& 0.8 & 0.0001 &2 & 53.815 & 31.835 & 2 & 53.815 & 31.835 & 1 & 56.430 & 33.248 & 3 & 8.395 & 11.674 & 1 & 56.430 & 33.248 \\
\bottomrule
			\end{tabular}}
		\end{center}
	\end{table}

	\begin{table}
		\scriptsize
		\setlength{\tabcolsep}{2pt}
		\caption{\label{tab:BCv1B} Results of $\overline{\mFB}$[+IPCs,TC-I,II] on B.}
		\begin{center}
			\begin{tabular*}{\textwidth}{@{\extracolsep{\fill}}rrrrrrrrrrrr}
				\toprule
				Group&$\#opt$&$\#opt_0$&$\#null-obj$&$lazy-cuts$&$user-cuts$&$grb-nodes$&$\%grb-gap$&$t_{sep}(s)$&$t(s)$&$paths$&$nodes$\\
				
				\midrule
				$G_1$  &   5 &          0 &          0 &      244.6 &     1872.0 &     1062.2 &       - &      0.000 &      0.321 &           29.8 &       84.0  \\
				
				$G_2$&      5 &          0 &          0 &     2641.8 &    22410.8 &    11273.6 &       - &      0.000 &     14.919 &          48.6 &      150.0 \\ 
				
				$G_3$&   4 &          0 &          0 &    18377.4 &   187050.8 &   239446.6 &      0.004 &      0.000 &    744.497 &         54.2 &      158.2   \\
				
				$G_4$&    4 &          0 &          0 &    30194.2 &   597104.8 &   299506.8 &     35.516 &      0.004 &   1329.738 &         61.6 &      199.0  \\
				
				$G_5$&     4 &          0 &          0 &    12752.4 &   225594.0 &   111251.8 &      0.002 &      0.001 &    771.869 &          72.0 &      240.0    \\
				
				$G_6$&     2 &          0 &          0 &    36533.8 &   764830.0 &   383725.6 &      7.311 &      0.003 &   2303.893 &         79.6 &      259.4  \\
				
				$G_7$&   4 &          0 &          0 &    22655.6 &    94175.2 &    94327.2 &     11.942 &      0.131 &    928.903 &         89.0 &      293.6  \\
				
				$G_8$& 0 &          0 &          0 &    16942.0 &   137056.0 &    68688.0 &     25.028 &      0.007 &   3601.425 &        157.0 &      400.0  \\
				\bottomrule
			\end{tabular*}
		\end{center}
	\end{table}
	
	\begin{table}
		\scriptsize
		\setlength{\tabcolsep}{2pt}
		\caption{\label{tab:BCv2B} Results of $\overline{\mFB}$[+A-GRC] on B.}
		\begin{center}
			\begin{tabular*}{\textwidth}{@{\extracolsep{\fill}}rrrrrrrrrrrr}
				\toprule
				Group&$\#opt$&$\#opt_0$&$\#null-obj$&$lazy-cuts$&$user-cuts$&$grb-nodes$&$\%grb-gap$&$t_{sep}(s)$&$t(s)$&$paths$&$nodes$\\
				
				\midrule
				$G_1$  &    5 &          0 &          0 &      244.6 &     1872.0 &     1062.2 &       - &      0.001 &      0.330 &          29.8 &       84.0 \\
				
				$G_2$&   5 &          0 &          0 &     2641.8 &    22410.8 &    11273.6 &       - &      0.008 &     15.177 &        48.6 &      150.0  \\ 
				
				$G_3$&   4 &          0 &          0 &    18082.8 &   184019.6 &   232729.8 &      0.004 &      0.017 &    745.751 &         54.2 &      158.2   \\
				
				$G_4$&   4 &          0 &          0 &    30352.8 &   601106.4 &   301581.2 &     35.516 &      0.333 &   1327.375 &         61.6 &      199.0 \\
				
				$G_5$&    4 &          0 &          0 &    12724.4 &   224892.4 &   110933.6 &      0.002 &      0.278 &    771.780 &         72.0 &      240.0   \\
				
				$G_6$&    2 &          0 &          0 &    36356.4 &   762680.4 &   382671.8 &      7.311 &      3.911 &   2305.253 &         79.6 &      259.4   \\
				
				$G_7$&    4 &          0 &          0 &    22714.0 &    94589.4 &    94741.4 &     11.942 &      0.619 &    927.038 &       89.0 &      293.6  \\
				
				$G_8$& 0 &          0 &          0 &    17106.0 &   138142.0 &    69231.0 &     25.028 &      0.038 &   3600.612 &       157.0 &      400.0 \\
				\bottomrule
			\end{tabular*}
		\end{center}
	\end{table}
	
	\begin{table}
		\scriptsize
		\setlength{\tabcolsep}{2pt}
		\caption{\label{tab:BCv2cB} Results of $\overline{\mFB}$[+A-GRC] (with Cliquer) on B.}
		\begin{center}
			\begin{tabular*}{\textwidth}{@{\extracolsep{\fill}}rrrrrrrrrrrr}
				\toprule
				Group&$\#opt$&$\#opt_0$&$\#null-obj$&$lazy-cuts$&$user-cuts$&$grb-nodes$&$\%grb-gap$&$t_{sep}(s)$&$t(s)$&$paths$&$nodes$\\
				
				\midrule
				$G_1$  &    5 &          0 &          0 &      186.2 &     1335.8 &      700.8 &      0.000 &      2.279 &      2.586 &        29.8 &       84.0 \\
				
				$G_2$&     5 &          0 &          0 &     2470.4 &    21258.2 &    10353.4 &      0.000 &     12.127 &     27.928  &       48.6 &      150.0  \\ 
				
				$G_3$&    4 &          0 &          0 &     9001.8 &   112048.4 &    64212.2 &      0.001 &    756.713 &    912.710 &        54.2 &      158.2  \\
				
				$G_4$&   3 &          0 &          0 &    30889.0 &   727310.6 &   363231.2 &      3.343 &    578.244 &   2160.092 &    64.8 &      204.0 \\
				
				$G_5$&   4 &          0 &          0 &    12165.6 &   208905.4 &   103215.0 &      7.594 &    383.077 &   1136.613 &  70.0 &      226.8\\
				
				$G_6$&     2 &          0 &          0 &    32272.8 &   694519.4 &   345330.2 &      5.571 &    457.788 &   2318.168 &  78.8 &      256.8\\
				
				$G_7$&   1 &          0 &          0 &    25766.4 &   310860.4 &   152475.4 &     30.623 &    471.737 &   1482.345 &    82.2 &      247.4  \\
				
				$G_8$&   0 &          0 &          0 &    17448.0 &   143870.0 &    72117.0 &     25.028 &      0.018 &   3601.877 &   157.0 &      400.0 \\
				\bottomrule
			\end{tabular*}
		\end{center}
	\end{table}

	\begin{table}
		\scriptsize
		\setlength{\tabcolsep}{2pt}
		\caption{\label{tab:BCv3B} Results of $\overline{\mFB}$[+RC,$\pm,-,+$] on B.}
		\begin{center}
			\begin{tabular*}{\textwidth}{@{\extracolsep{\fill}}rrrrrrrrrrrr}
				\toprule
				Group&$\#opt$&$\#opt_0$&$\#null-obj$&$lazy-cuts$&$user-cuts$&$grb-nodes$&$\%grb-gap$&$t_{sep}(s)$&$t(s)$&$paths$&$nodes$\\
				
				\midrule
				$G_1$  &   5 &          0 &          0 &      210.8 &     1161.2 &      711.4 &       - &      0.142 &      0.464 &          29.8 &       84.0  \\
				
				$G_2$&  5 &          0 &          0 &     1961.6 &    12376.2 &     6371.2 &       -&      2.873 &     10.714 &         48.6 &      150.0  \\ 
				
				$G_3$&    4 &          0 &          0 &    11380.0 &   134140.0 &   115727.8 &      0.004 &    451.051 &    768.673 &        54.2 &      158.2    \\
				
				$G_4$&  4 &          0 &          0 &    26279.0 &   490876.0 &   245363.8 &     56.838 &     16.142 &   1132.501 &       63.2 &      194.8  \\
				
				$G_5$&     5 &          0 &          0 &     6642.0 &    83990.6 &    41702.0 &       - &     10.384 &    154.665 &        71.8 &      240.0   \\
				
				$G_6$&       2 &          0 &          0 &    26638.4 &   513585.4 &   255848.0 &      7.781 &     30.571 &   1280.014 &          78.4 &      258.6   \\
				
				$G_7$& 2 &          0 &          1 &    24682.2 &   125797.8 &   126250.6 &     14.338 &     37.284 &   1108.812 &          70.6 &      225.0\\
				
				$G_8$&   0 &          0 &          0 &    17106.0 &   138452.0 &    69386.0 &     25.028 &      0.017 &   3600.864 &     157.0 &      400.0  \\
				\bottomrule
			\end{tabular*}
		\end{center}
	\end{table}
	
	\begin{table}
		\scriptsize
		\setlength{\tabcolsep}{2pt}
		\caption{\label{tab:BCv3cB} Results of $\overline{\mFB}$[+RC,$\pm,-,+$] (with cliquer) on B.}
		\begin{center}
			\begin{tabular*}{\textwidth}{@{\extracolsep{\fill}}rrrrrrrrrrrr}
				\toprule
				Group&$\#opt$&$\#opt_0$&$\#null-obj$&$lazy-cuts$&$user-cuts$&$grb-nodes$&$\%grb-gap$&$t_{sep}(s)$&$t(s)$&$paths$&$nodes$\\
				
				\midrule
				$G_1$  &  5 &          0 &          0 &      186.2 &     1335.8 &      700.8 &      0.000 &      2.270 &      2.572  &       29.8 &       84.0  \\
				
				$G_2$&   5 &          0 &          0 &     2470.4 &    21258.2 &    10353.4 &      0.000 &     12.078 &     27.780  &       48.6 &      150.0 \\ 
				
				$G_3$&    4 &          0 &          0 &     8796.2 &   110730.6 &    63323.6 &      0.001 &    755.503 &    910.539 &        54.2 &      158.2  \\
				
				$G_4$&   3 &          0 &          0 &    30361.6 &   719751.6 &   359443.2 &      3.343 &    581.221 &   2161.194 &   64.8 &      204.0   \\
				
				$G_5$&    4 &          0 &          0 &    12451.0 &   216816.6 &   107047.8 &      7.594 &    380.555 &   1131.769 &    70.0 &      226.8 \\
				
				$G_6$&   2 &          0 &          0 &    32209.4 &   695198.6 &   345643.2 &      5.571 &    462.613 &   2326.343 &   78.8 &      256.8\\
				
				$G_7$&  1 &          0 &          0 &    25767.0 &   311138.0 &   152600.4 &     30.623 &    470.496 &   1481.438 &   82.2 &      247.4\\
				
				$G_8$&   0 &          0 &          0 &    17448.0 &   142922.0 &    71643.0 &     25.028 &      0.156 &   3601.107 &  157.0 &      400.0  \\
				\bottomrule
			\end{tabular*}
		\end{center}
	\end{table}
	
	\begin{table}
		\scriptsize
		\setlength{\tabcolsep}{2pt}
		\caption{\label{tab:compareLBB} Comparison of lower bounds on set B.}
		\begin{center}
			\resizebox{\textwidth}{!}{
			\begin{tabular}{rrrrrrrrrrrrrrrrrrr}
				
				\toprule
				Group& \multicolumn{3}{c}{$\overline{L\mFB}$} & \multicolumn{3}{c}{$\overline{L\mFB}$[+IPCs,TC-I,II]} & \multicolumn{3}{c}{$\overline{L\mFB}$[+A-GRC]} &  \multicolumn{3}{c}{$\overline{L\mFB}$[+A-GRC](C)} & \multicolumn{3}{c}{$\overline{L\mFB}$[+RC,$\pm,-,+$]} & \multicolumn{3}{c}{$\overline{L\mFB}$[+RC,$\pm,-,+$](C)}\\
				\cmidrule{2-4}	\cmidrule{5-7} \cmidrule{8-10} \cmidrule{11-13} \cmidrule{14-16} \cmidrule{17-19} 			
				&{\fontsize{6.5}{8}$\#bs$}&{\fontsize{6.5}{8} $\%obj$}&{\fontsize{5.5}{7} $\%best-gap$}&{\fontsize{6.5}{8}$\#bs$}&{\fontsize{6.5}{8} $\%obj$}&{\fontsize{5.5}{7} $\%best-gap$}&{\fontsize{6.5}{8}$\#bs$}&{\fontsize{6.5}{8} $\%obj$}&{\fontsize{5.5}{7} $\%best-gap$}&{\fontsize{6.5}{8}$\#bs$}&{\fontsize{6.5}{8} $\%obj$}&{\fontsize{5.5}{7} $\%best-gap$}&{\fontsize{6.5}{8}$\#bs$}&{\fontsize{6.5}{8} $\%obj$}&{\fontsize{5.5}{7} $\%best-gap$}&{\fontsize{6.5}{8}$\#bs$}&{\fontsize{6.5}{8} $\%obj$}&{\fontsize{5.5}{7} $\%best-gap$}\\
				
				\midrule
				$G_1$ &1 & 1.195 & 1.196 & 0 & 292.438 & 117.820 & 0 & 292.438 & 117.820 & 0 & 292.438 & 117.820 & 0 & 292.438 & 117.820 & 0 & 292.438 & 117.820\\
				$G_2$ &2 & 0.005 & 0.005 & 0 & 94.314 & 16.623 & 0 & 94.314 & 16.623 & 0 & 94.314 & 16.623 & 0 & 94.314 & 16.623 & 0 & 94.314 & 16.623\\
				$G_3$ &2 & 1.143 & 1.143 & 0 & 96.272 & 16.988 & 0 & 96.272 & 16.988 & 0 & 96.272 & 16.988 & 0 & 96.272 & 16.988 & 0 & 96.272 & 16.988 \\
				$G_4$ & 2 & 2.868 & 2.867 & 0 & 466.667 & 16.665 & 0 & 466.667 & 16.665 & 0 & 466.667 & 16.665 & 0 & 466.667 & 16.665 & 0 & 466.667 & 16.665\\
				$G_5$ & 1 & 0.008 & 0.008 & 0 & 518.569 & 23.713 & 0 & 518.569 & 23.713 & 0 & 518.569 & 23.713 & 0 & 518.569 & 23.713 & 0 & 518.569 & 23.713\\
				$G_6$ & 1 & 3.785 & 3.784 & 0 & 101.702 & 20.782 & 0 & 101.702 & 20.782 & 0 & 101.702 & 20.782 & 0 & 101.702 & 20.782 & 0 & 101.702 & 20.782 \\
\bottomrule
			\end{tabular}}
		\end{center}
	\end{table}
	
	\begin{table}
		\scriptsize
		\setlength{\tabcolsep}{2pt}
		\caption{\label{tab:compareBCB} Comparison of BC variants on set B.}
		\begin{center}
		\resizebox{\textwidth}{!}{	\begin{tabular}{rrrrrrrrrrrrrrrr}
				
				\toprule
				Group&\multicolumn{3}{c}{$\overline{\mFB}$[+IPCs,TC-I,II]} & \multicolumn{3}{c}{$\overline{\mFB}$[+A-GRC]} &  \multicolumn{3}{c}{$\overline{\mFB}$[+A-GRC](C)} & \multicolumn{3}{c}{$\overline{\mFB}$[+RC,$\pm,-,+$]} & \multicolumn{3}{c}{$\overline{\mFB}$[+RC,$\pm,-,+$](C)}\\
				\cmidrule{2-4}	\cmidrule{5-7} \cmidrule{8-10} \cmidrule{11-13} \cmidrule{14-16} 	
				&{\fontsize{6.5}{8}$\#bs$}&{\fontsize{6.5}{8} $\%obj$}&{\fontsize{5.5}{7} $\%best-gap$}&{\fontsize{6.5}{8}$\#bs$}&{\fontsize{6.5}{8} $\%obj$}&{\fontsize{5.5}{7} $\%best-gap$}&{\fontsize{6.5}{8}$\#bs$}&{\fontsize{6.5}{8} $\%obj$}&{\fontsize{5.5}{7} $\%best-gap$}&{\fontsize{6.5}{8}$\#bs$}&{\fontsize{6.5}{8} $\%obj$}&{\fontsize{5.5}{7} $\%best-gap$}&{\fontsize{6.5}{8}$\#bs$}&{\fontsize{6.5}{8} $\%obj$}&{\fontsize{5.5}{7} $\%best-gap$}\\
				
				\midrule
				$G_1$ &5 & 0.000 & 0.000 & 5 & 0.000 & 0.000 & 5 & 0.000 & 0.000 & 5 & 0.000 & 0.000 & 5 & 0.000 & 0.000\\
				$G_2$ &5 & 0.000 & 0.000 & 5 & 0.000 & 0.000 & 5 & 0.000 & 0.000 & 5 & 0.000 & 0.000 & 5 & 0.000 & 0.000 \\
				$G_3$ &5 & 0.000 & 0.000 & 5 & 0.000 & 0.000 & 5 & 0.000 & 0.000 & 5 & 0.000 & 0.000 & 5 & 0.000 & 0.000\\
				$G_4$ & 4 & 2.509 & 3.222 & 4 & 2.509 & 3.222 & 4 & 0.000 & 0.000 & 4 & 4.727 & 6.875 & 4 & 0.000 & 0.000 \\
				$G_5$ &4 & 0.000 & 0.000 & 4 & 0.000 & 0.000 & 4 & 5.824 & 7.594 & 5 & 0.000 & 0.000 & 4 & 5.824 & 7.594 \\
				$G_6$ &4 & 0.154 & 0.160 & 4 & 0.154 & 0.160 & 4 & 1.168 & 1.266 & 4 & 0.463 & 0.453 & 4 & 1.168 & 1.266\\
				$G_7$ &5 & 0.000 & 0.000 & 5 & 0.000 & 0.000 & 1 & 18.686 & 28.208 & 3 & 30.494 & 4.184 & 1 & 18.686 & 28.208\\
				$G_8$ & 1 & 0.000 & 0.000 & 1 & 0.000 & 0.000 & 1 & 0.000 & 0.000 & 1 & 0.000 & 0.000 & 1 & 0.000 & 0.000 \\
				\midrule
				&\textbf{4.125}&	\textbf{0.333}&	\textbf{0.423}&	\textbf{4.125}&	\textbf{0.333}&	\textbf{0.423}&	3.625&	3.210&	4.633&	4&	4.460&	1.439	&3.625&	3.210&	4.633\\

				\bottomrule
			\end{tabular}}
		\end{center}
	\end{table}
	
	\begin{table}
		\scriptsize
		\setlength{\tabcolsep}{2pt}
		\caption{\label{tab:BCv1C} Results of $\overline{\mFB}$[+IPCs,TC-I,II] on C.}
		\begin{center}
			%
		\end{center}
	\end{table}
	
	\begin{table}
		\scriptsize
		\setlength{\tabcolsep}{2pt}
		\caption{\label{tab:BCv2C} Results of $\overline{\mFB}$[+A-GRC] on C.}
		\begin{center}
			%
		\end{center}
	\end{table}
	
	\begin{table}
		\scriptsize
		\setlength{\tabcolsep}{2pt}
		\caption{\label{tab:BCv2cC} Results of $\overline{\mFB}$[+A-GRC] (with cliquer) on C.}
		\begin{center}
			%
		\end{center}
	\end{table}
	
	\begin{table}
		\scriptsize
		\setlength{\tabcolsep}{2pt}
		\caption{\label{tab:BCv3C} Results of $\overline{\mFB}$[+RC,$\pm,-,+$] on C.}
		\begin{center}
			%
		\end{center}
	\end{table}
	
	\begin{table}
		\scriptsize
		\setlength{\tabcolsep}{2pt}
		\caption{\label{tab:BCv3cC} Results of $\overline{\mFB}$[+RC,$\pm,-,+$] (with cliquer) on C.}
		\begin{center}
			%
		\end{center}
	\end{table}
	
	\begin{table}
		\tiny
		\setlength{\tabcolsep}{2pt}
		\caption{\label{tab:compareLBC} Comparison of lower bounds on set C.}
		\begin{center}
			\resizebox{\textwidth}{!}{
			%
}
		\end{center}
	\end{table}
	
	\begin{table}
		\tiny
		\setlength{\tabcolsep}{2pt}
		\caption{\label{tab:compareBCC} Comparison of BC variants on set C.}
		\begin{center}
			%
		\end{center}
\end{table}

\end{document}